%% file: report.tmp.tex
\newif\ifdraft\draftfalse  
\newif\iffull\fullfalse   
\newif\ifbackref\backreffalse 
\newif\ifsooner\soonerfalse
\newif\iflater\laterfalse
\newif\iftr\trtrue       
\makeatletter \@input{texdirectives} \makeatother

\iftr
\documentclass[11pt,twocolumn,titlepage]{article}



\else
\documentclass[11pt,twocolumn]{article}
\fi

\newcommand{\EG}{e.g.~}
\newcommand{\IE}{i.e.~}
\newcommand{\ETAL}{et al.~}
\newcommand{\ETC}{etc.~}

\usepackage{supertabular}
\usepackage{ifthen}
\usepackage{alltt}
\usepackage{color}

\usepackage{cite}
\usepackage{fontenc}[T1]
\usepackage{xcolor}
\definecolor{darkblue}{rgb}{0.0, 0.0, 0.55}

\usepackage[american]{babel}
\usepackage[T1]{fontenc}
\usepackage[utf8]{inputenc}
\newcommand{\border}{1.5cm}
\usepackage[left=\border,right=\border,top=\border,bottom=\border,verbose=false]{geometry}

\usepackage{amsmath}
\usepackage{amssymb}
\usepackage{amsthm}
\usepackage{mathtools}
\usepackage{mathrsfs}

\usepackage{xcolor}
\usepackage[normalem]{ulem}

\usepackage[
    pdftex,%
    pdfpagelabels,%
    colorlinks,%
    linkcolor=darkblue,%
    citecolor=darkblue,%
    filecolor=darkblue,%
    urlcolor=darkblue%
]{hyperref}

\def\Snospace~{\S{}}
\addto\extrasamerican{
}

\newtheorem{lem}{Lemma}

\newtheorem{thm}{Theorem}

\include{rfj.ott}

\definecolor{darkblue}{rgb}{0,0.1,0.5}
\definecolor{dkblue}{rgb}{0,0.1,0.5}
\definecolor{dkgreen}{rgb}{0,0.4,0}
\definecolor{dkred}{rgb}{0.6,0,0}
\definecolor{dkpurple}{rgb}{0.7,0,0.4}
\definecolor{olive}{rgb}{0.5, 0.5, 0.0}
\definecolor{teal}{rgb}{0.0,0.5,0.5}
\definecolor{orchid}{rgb}{0.85,0.44,0.84}

\long\def\comment#1{}

\newcommand{\comm}[3]{\ifdraft{\color{#1}[#2: #3]}\fi}
\newcommand{\bcp}[1]{\comm{dkpurple}{BCP}{#1}}

\newcommand{\ch}[1]{\comm{dkgreen}{CH}{#1}} 
\newcommand{\chrev}[1]{{\ifdraft\textcolor{dkgreen}{\fi#1\ifdraft}\fi}} 
\newcommand{\yj}[1]{\comm{orchid}{YJ}{#1}} 

\begin{document}

\title{\bf Towards a Fully Abstract Compiler Using Micro-Policies\\[1ex]
  \Large --- Secure Compilation for Mutually Distrustful Components ---}



\iftr
\author{
  Yannis Juglaret\textsuperscript{1,2}\qquad
  C\u{a}t\u{a}lin Hri\c{t}cu\textsuperscript{1}\qquad 
  Arthur Azevedo de Amorim\textsuperscript{3}\\[2ex] 
  Benjamin C. Pierce\textsuperscript{3} \qquad 
  Antal Spector-Zabusky\textsuperscript{3} \qquad 
  Andrew Tolmach\textsuperscript{4}\\[2em] 
  \textsuperscript{1}Inria Paris\qquad
  \textsuperscript{2}Universit\'{e} Paris Diderot (Paris 7)\\[2ex]
  \textsuperscript{3}University of Pennsylvania\qquad
  \textsuperscript{4}Portland State University\\[2em]
}
\date{Technical Report\\[0.7em]
      initial version: August 21, 2015\\[0.7em]
      last revised: \today}
\fi

\iftr
\thispagestyle{empty}
\fi

\maketitle

\newgeometry{left=4.5cm,right=4.5cm}
\begin{abstract}
  Secure compilation prevents all low-level attacks on compiled code
  and allows for sound reasoning about security in the source language.
  In this work we propose a new attacker model for secure compilation
  that extends the well-known notion of full abstraction to ensure
  protection for mutually distrustful components.
  We devise a compiler chain (compiler, linker, and loader) and a
  novel security monitor that together defend against this strong
  attacker model.
  The monitor is implemented using a recently proposed, generic
  tag-based protection framework called micro-policies, which comes
  with hardware support for efficient caching and with a formal
  verification methodology.
  Our monitor protects the abstractions of a simple object-oriented
  language---class isolation, the method call discipline, and type
  safety---against arbitrary low-level attackers.
\end{abstract}
\restoregeometry

\onecolumn
\tableofcontents
\twocolumn

\pagestyle{plain}


\iftr
\section{Introduction}
\fi

\iftr
\subsection{General Context}
\fi




%
In this work we study compiled partial programs evolving within a low-level
environment, with which they can interact.
Such interaction is useful --- think of high-level programs performing
low-level library calls, or of a browser interacting with native code
that was sent over the internet~\cite{DBLP:conf/icics/DaiSYL12,
DBLP:journals/cacm/YeeSDCMOONF10} --- but also dangerous: parts of the
environment could be malicious or compromised and try to compromise
the program as well~\cite{DBLP:conf/fosad/Erlingsson07,
DBLP:conf/icics/DaiSYL12,DBLP:journals/cacm/YeeSDCMOONF10}.
Low-level libraries written in C or in C++ can be
vulnerable to control hijacking attacks~\cite{Szekeres2013,
DBLP:conf/fosad/Erlingsson07} and be taken over by a remote attacker.
When the environment can't be trusted, it is a major concern to ensure
the security of running programs.

%
With today's compilers, low-level attackers~\cite{
DBLP:conf/fosad/Erlingsson07} can circumvent high-level
abstractions~\cite{abadi_protection98,DBLP:journals/tcs/Kennedy06} and
are thus much more powerful than high-level attackers, which means
that the security reasoning has to be done at the lowest level, which
it is extremely difficult.
An alternative is to build a {\em secure compiler} that ensures that
low- and high-level attackers have exactly the same power, allowing
for easier, source-level security reasoning~\cite{abadi_aslr12,
DBLP:conf/csfw/JagadeesanPRR11,patrignani2014secure,
DBLP:conf/popl/FournetSCDSL13}.
%
Formally, the notion of secure compilation is usually expressed
as \emph{full abstraction} of the
translation~\cite{abadi_protection98}.
Full abstraction is a much stronger property than just compiler
correctness~\cite{leroy09:compcert}.

%
Secure compilation is, however, very hard to achieve in practice.
Efficiency, which is crucial for broad adoption~\cite{Szekeres2013},
is the main challenge.
Another concern is transparency.
%
While we want to constrain the power of low-level attackers, the
constraints we set should be relaxed enough that there is a way for all
benign low-level environments to respect them.
%
If we are not transparent enough, the partial program might be
prevented from properly interacting with its environment
(e.g.\ the low-level libraries it requires).

%
For a compiler targeting machine code, which lacks structure and
checks, a typical low-level attacker has write access to the whole
memory, and can redirect control flow to any location in
memory~\cite{DBLP:conf/fosad/Erlingsson07}.
Techniques have been developed to deal with such powerful
attackers, in particular
involving randomization~\cite{abadi_aslr12} and binary code
rewriting~\cite{ErlingssonAVBN06,MorrisettTTTG12}.
The first ones only offer weak probabilistic guarantees; as a
consequence, address space layout randomization~\cite{abadi_aslr12} is
routinely circumvented in practical
attacks~\cite{DBLP:conf/sp/SnowMDDLS13,DBLP:conf/sp/EvansFGOTSSRO15}.
The second ones add extra software checks which often come at a high
performance cost.

%
Using additional protection in the hardware can result in secure
compilers with strong guarantees~\cite{patrignani2014secure}, without
sacrificing efficiency or transparency.
Because updating hardware is expensive and hardware adoption
takes decades, the need for generic protection mechanisms that can fit
with ever-evolving security requirements has emerged.
Active research in the domain includes \emph{capability
machines}~\cite{DBLP:conf/asplos/CarterKD94,cheri2012,
cheri_oakland2015} and \emph{tag-based architectures}~\cite{
ieee_hst2013,PicoCoq2013,pump_asplos2015,micropolicies2015}.
%
In this work, we use a generic tag-based protection mechanism called
  \emph{micro-policies}~\cite{pump_asplos2015, micropolicies2015} as
  the target of a secure compiler.  

%
Micro-policies provide instruction-level
monitoring based on fine-grained metadata tags.
In a micro-policy machine, every word of data is augmented with a
word-sized tag, and a hardware-accelerated monitor propagates
these tags every time a machine instruction gets executed.
Micro-policies can be described as a combination of
software-defined rules and monitor services.
The rules define how the monitor will perform tag propagation
instruction-wise, while the services allow for direct interaction
between the running code and the monitor.
This mechanism comes with an efficient hardware implementation built
on top of a RISC processor~\cite{pump_asplos2015} as well as a mechanized
metatheory~\cite{micropolicies2015}, and has already been used to
enforce a variety of security policies~\cite{pump_asplos2015,
micropolicies2015}.

\iftr
\subsection{Research Problem}
\fi


Recent work~\cite{DBLP:conf/csfw/AgtenSJP12,patrignani2014secure}
has illustrated how \emph{protected module architectures} --- a class
of hardware architectures featuring coarse-grained isolation
mechanisms~\cite{DBLP:conf/isca/McKeenABRSSS13,
DBLP:conf/isca/HoekstraLPPC13,DBLP:conf/ccs/StrackxP12} --- can help
in devising a fully abstract compilation scheme for a Java-like
language.
This scheme assumes the compiler knows which components
in the program can be trusted and which ones cannot, and protects the
trusted components from the distrusted ones by isolating them in a
protected module.

%
%
This kind of protection is only appropriate when all the components we
want to protect can be trusted, for example because they have been
verified~\cite{DBLP:conf/popl/Agten0P15}.
Accounting for the cases in which this is not possible, we present and
adopt a stronger attacker model of \emph{mutual distrust}: in this
setting a secure compiler should protect each component from every
other component, so that whatever the compromise scenario
may be, uncompromised components always get protected from the
compromised ones.

The main questions we address in this work are:
(1) can we build a fully abstract compiler to a micro-policy machine?
and (2) can we support a stronger attacker model by
protecting mutually distrustful components against each other?

We are the first to work on question 1, and among
  the first to study question 2:
Micro-policies are a recent hardware mechanism~\cite{
pump_asplos2015,micropolicies2015} that is flexible and fine-grained
enough to allow building a secure compiler against this strong
attacker model.
In independent parallel work~\cite{patrignani_thesis,PatrignaniDP15},
Patrignani et al.\ are trying to extend their previous
results~\cite{patrignani2014secure} to answer question 2 using
different mechanisms (\EG multiple protected modules and
randomization).
Related work is further discussed in \autoref{sec:related}.

\iftr
\subsection{Our Contribution}
\fi

  In this work we propose a new attacker model for secure compilation
  that extends the well-known notion of full abstraction to ensure
  protection for mutually distrustful components (\autoref{sec:attacker-model}).
  We devise a secure compilation solution (\autoref{sec:overview}) for
  a simple object-oriented language (\autoref{sec:source}) that
  defends against this strong attacker model.
  Our solution includes a simple compiler chain (compiler,
  linker, and loader; \autoref{sec:compiler}) and a novel micro-policy
  (\autoref{sec:micro-policy}) that protects the abstractions of our
  simple language---class isolation, the method call discipline, and
  type safety---against arbitrary low-level attackers.
  Enforcing a method call discipline and type safety using a
  micro-policy is novel and constitutes a contribution of independent
  interest.






\iflater
\else
We have started proving that our compiler is secure, but since
that proof is not yet finished, we do not present it in the report.
\fi
Section \ref{sec:efficiency} explains why we have good hopes in the
efficiency and transparency of our solution for the protection of
realistic programs.
We also discuss ideas for mitigation when our mechanism is not
transparent enough.
However, in both cases gathering evidence through experiments to
confirm our hopes is left for future work.

\subsection{Other Insights}

\bcp{Suggest postponing these discussions to the conclusions.  Putting them
  here slows down the main story.}

\ch{Any important insights we are still missing?}%
\ch{{\em Completely/mostly standard} un-optimizing compiler with
  separate compilation (not whole program compilation) + micro-policy?
  See discussion in 6.1}


Throughout this work, we reasoned a lot about abstractions.
One insight we gained is that even very simple high-level languages are
much more abstract than one would naively expect.
Moreover, we learned that some abstractions --- such as
functional purity --- are impossible to efficiently enforce
dynamically.

We also needed to extend the current hardware and formalism of
micro-policies (\autoref{sec:target}) in order to achieve our
challenging security goal.
We needed two kinds of extensions: some only ease micro-policy
writing, while the others increase the power to the monitoring mechanism.
The first ones require a policy compiler, allowing an easier
specification for complex micro-policies, which can then still run on
the current hardware.
The second ones require actual hardware extensions.
Both of these extensions keep the spirit of micro-policies unchanged:
Rules, in particular, are still specified as a mapping from tags to
tags.

Finally, as we mention in \autoref{sec:i2t}, we were able to provide
almost all security at the micro-policy level rather than the compiler
level.
This is very encouraging because it means that we might be able to
provide full abstraction for complex compilers that already exist,
using micro-policies while providing very little change to the
compiler itself.

\section{Stronger Attacker Model for Secure Compilation
  of Mutually Distrustful Components}
\label{sec:attacker-model}

Previous work on secure compilation~\cite{abadi_aslr12,
DBLP:conf/csfw/JagadeesanPRR11,patrignani2014secure,
DBLP:conf/popl/FournetSCDSL13} targets a property called
full abstraction~\cite{abadi_protection98}.
This section presents full abstraction
(\autoref{sec:fullabstraction}), motivates why it is not
enough in the context of mutually distrustful components
(\autoref{sec:not-enough}), and introduces a stronger attacker model
for this purpose (\autoref{sec:mutual-distrust}).

\subsection{Full Abstraction}
\label{sec:fullabstraction}

Full abstraction is a property of compilers that talks about the
observable behaviors of partial programs evolving in a context.
When we use full abstraction for security purposes, we will think
of contexts as \emph{attackers} trying to learn the partial program's
secrets, or to break its internal invariants.
Full abstraction relates the observational power of low-level contexts
to that of high-level contexts.
Hence, when a compiler achieves full abstraction, low-level attackers
can be modeled as high-level ones, which makes reasoning about the
security of programs much easier:
Because they are built using source-level abstractions, high-level
attackers have more structure and their interaction with the program
is limited to that allowed by the semantics of the source language.

In order to state full abstraction formally, one first has to provide
definitions for partial programs, contexts, and observable behaviors
both in the high- and the low-level.
Partial programs are similar to usual programs; but they
could still be missing some elements---\EG
external libraries---before they can be executed.
The usual, complete programs can be seen as a particular case of
partial programs which have no missing elements, and are thus ready
for execution.
A context is usually defined as a partial program with a hole;
this hole can later be filled with a partial program in order to yield
a new partial program.
Finally, observable behaviors of complete programs can vary depending
on the language and may include, termination,
I/O during execution, final result value, 
or final memory state. 

The chosen definition for contexts will set the granularity at which
the attacker can operate.
%
Similarly, defining the observable behaviors of complete programs can affect
the observational power of the attacker in our formal model.
The attacker we want to protect the program against is the context
itself:
The definition we choose for observable behaviors should allow the context to produce an
observable action every time it has control, thus letting it convert
its knowledge into observable behaviors.
In our case, our source and target languages feature immediate program
termination constructs.
We can thus choose \emph{program termination} as an observable
behavior which models such strong observational power.
%

We denote high-level partial programs by $\ottnt{P}, \ottnt{Q}$, and
high-level contexts by $\ottnt{A}$.
We denote by $\ottnt{A}  \ottsym{[}  \ottnt{P}  \ottsym{]}$ the partial program obtained by inserting a
high-level partial program $\ottnt{P}$ in a high-level context $\ottnt{A}$.
We denote low-level partial programs by $\ottnt{p}, \ottnt{q}$, and
high-level contexts by $\ottnt{a}$.
We denote by $\ottnt{a}  \ottsym{[}  \ottnt{p}  \ottsym{]}$ the insertion of a low-level partial program
$\ottnt{p}$ in a low-level context $\ottnt{a}$.
Given a high-level partial program $\ottnt{P}$, we denote by $\ottnt{P}  \hspace{-0.35em}\downarrow$ the
low-level program obtained by compiling $\ottnt{P}$.
We denote the fact that two complete high-level programs $\ottnt{P}$
and $\ottnt{Q}$ have the same observable behavior by $\ottnt{P}  \sim_H  \ottnt{Q}$.
For two complete low-level programs $\ottnt{p}$ and $\ottnt{q}$, we denote
this by $\ottnt{p}  \sim_L  \ottnt{q}$.
With these notations, full abstraction of the compiler is stated as
\begin{multline*}
\ottsym{(}  \forall \, \ottnt{A}  \ottsym{,}  \ottnt{A}  \ottsym{[}  \ottnt{P}  \ottsym{]}  \sim_H  \ottnt{A}  \ottsym{[}  \ottnt{Q}  \ottsym{]}  \ottsym{)}  \iff  \ottsym{(}  \forall \, \ottnt{a}  \ottsym{,}  \ottnt{a}  \ottsym{[}  \ottnt{P}  \hspace{-0.35em}\downarrow  \ottsym{]}  \sim_L  \ottnt{a}  \ottsym{[}  \ottnt{Q}  \hspace{-0.35em}\downarrow  \ottsym{]}  \ottsym{)}
\end{multline*}
for all $\ottnt{P}$ and $\ottnt{Q}$.
Put into words, a compiler is fully abstract when any two
high-level partial programs $\ottnt{P}$ and $\ottnt{Q}$ behave the same in
every high-level context if and only if the compiled partial programs
$\ottnt{P}  \hspace{-0.35em}\downarrow$ and $\ottnt{Q}  \hspace{-0.35em}\downarrow$
behave the same in every low-level context.
In other words, a compiler is fully abstract when a low-level attacker
is able to distinguish between exactly the same programs as a
high-level attacker.

Intuitively, in the definition of full abstraction the trusted
compiled program ($\ottnt{P}  \hspace{-0.35em}\downarrow$ or $\ottnt{Q}  \hspace{-0.35em}\downarrow$) is protected from the
untrusted low-level context ($\ottnt{a}$) in a way that the context cannot
cause more harm to the program than a high-level context ($\ottnt{A}$)
already could.
This static separation between the trusted compiled program and the
context is in practice chosen by the user and communicated to the
compiler chain, which can insert a single protection barrier between
between the two.
In particular, in the definition of full abstraction the
compiler is only ever invoked for the protected program ($\ottnt{P}  \hspace{-0.35em}\downarrow$ or
$\ottnt{Q}  \hspace{-0.35em}\downarrow$), and can use this fact to its advantage, \EG to add dynamic checks.
Moreover, the definition of $\ottnt{a}  \ottsym{[}  \ottnt{p}  \ottsym{]}$ (low-level linking) can insert
a dynamic protection barrier between the trusted $p$ and the untrusted $a$.
For instance, Patrignani~\ETAL\cite{patrignani2014secure} built a
fully abstract compilation scheme targeting protected module
architectures by putting the compiled program into a protected part of
the memory (the protected module) and giving only unprotected memory
to the context.
A {\em single} dynamic protection barrier is actually enough to
enforce the full abstraction attacker model.

\subsection{Limitations of Full Abstraction}
\label{sec:not-enough}

\ch{Why is full abstraction not enough for multiple distrustful
  compartments}

We study languages for which programs can be decomposed into
\emph{components}.
Real-world languages have good candidates for such a notion of
components: depending on the granularity we target, they could be
packages, modules, or classes.
Our compiler is such that source components are \emph{separately compilable}
program units, and compilation maps source-level components to
target-level components.

When using a fully abstract compiler in the presence of multiple
components, the user has a
choice whether a component written in the high-level language is
trusted, in which case it is considered part of the program, or
untrusted, in which case it is considered part of the context.
If it is untrusted, the component can as well be compiled with an
insecure compiler, since anyway the fully abstract compiler only
provides security to components that are on the good side of the
protection barrier.
If the program includes components written in the low-level language,
\EG for efficiency reasons, then the user has generally no choice but
to consider these components untrusted.
Because of the way full abstraction is stated, low-level components
that are not the result of the compiler cannot be part of the trusted
high-level program, unless they have at least a high-level equivalent
(we discuss this idea in \autoref{sec:mutual-distrust}).

\begin{figure}
\includegraphics[trim=0.9cm 22cm 5.5cm
1.1cm,clip,width=\columnwidth]{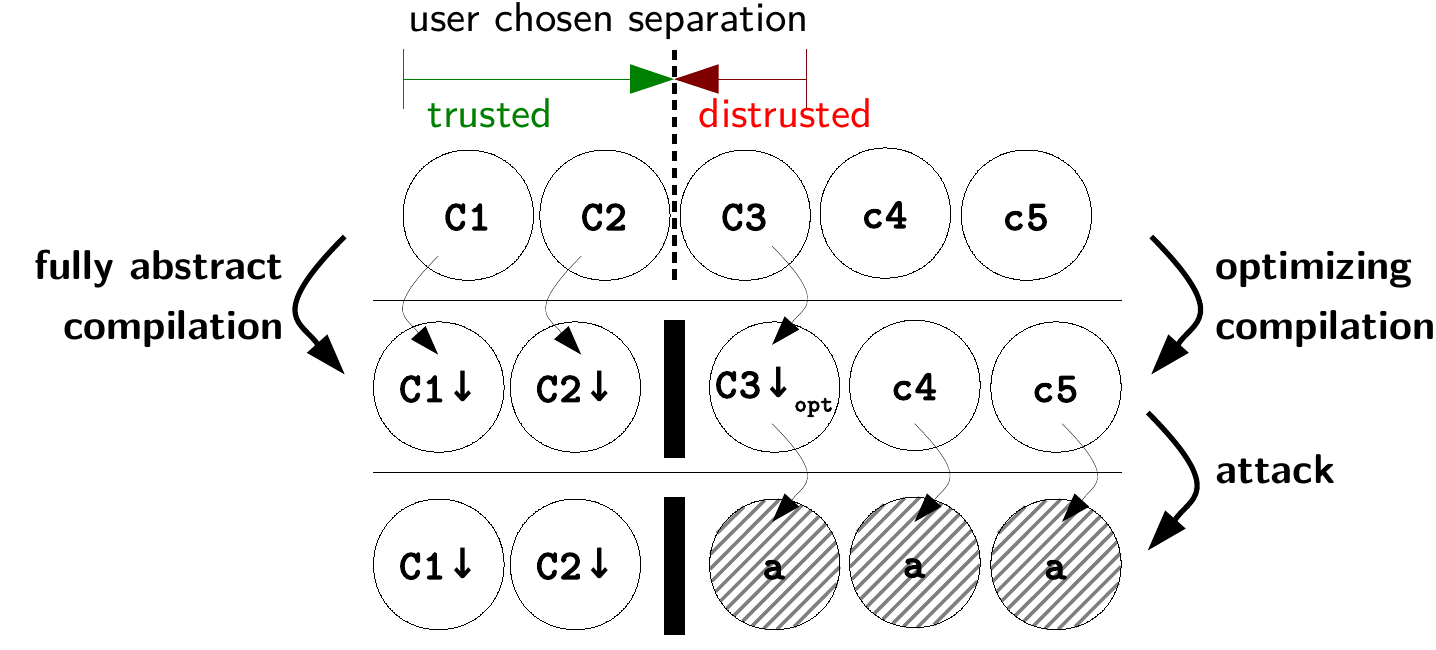}
\caption{Full abstraction for multiple components}
\label{fig:multi-component-full-abstraction}
\end{figure}

Figure~\ref{fig:multi-component-full-abstraction} graphically
illustrates how full abstraction could be applied in a multi-component
setting.
Components {\bf C1}, {\bf C2}, and {\bf C3} are written in the
high-level language, while {\bf c4} and {\bf c5} are written in
the low-level one.
Suppose the user chooses to trust {\bf C1} and {\bf C2} and not to
trust {\bf C3}, then the compiler will introduce a single barrier
protecting {\bf C1} and {\bf C2} from all the other components.


%
There are two assumptions on the attacker model when we take full
abstraction as a characterization of secure compilation:
{\em the user correctly identifies trusted and untrusted components so that
(1) trusted components need not be protected from each other, and
%
%
(2) untrusted components need no protection whatsoever.}
%
%
%
We argue that there are common cases in which the absolute, binary trust notion
implied by full abstraction is too limiting (\EG there is no way to
achieve all the user's security goals), and for which a stronger
attacker model protecting mutually distrustful components is needed.

Assumption (1) is only realistic if all trusted components are memory
safe~\cite{arthur-alpha-draft} and do not exhibit C-style undefined behaviors.
Only when all trusted components have a well-defined semantics in
the high-level language is a fully abstract compiler required to
preserve this semantics at the low level.
Memory safety for the trusted components may follow either from the
fact that the high-level language is memory safe as a whole or that
the components have been verified to be memory
safe~\cite{DBLP:conf/popl/Agten0P15}.
In the typical case of unverified C code, however, assumption (1) can
be unrealistically strong, and the user cannot be realistically
expected to decide which components are memory safe.
If he makes the wrong choice all bets are off for security,
a fully abstract compiler can produce code in which a control
hijacking attack~\cite{Szekeres2013,DBLP:conf/fosad/Erlingsson07} in
one trusted component can take over all the rest.
While we are not aware of any fully abstract compiler for unverified C,
we argue that if one focuses solely on achieving the full abstraction
property, such a compiler could potentially be as insecure in practice
as standard compilers.
%

Even in cases where assumption (1) is acceptable, assumption (2) is
still a very strong one.
In particular, since components written in the low-level language
cannot get protection, every security-critical component would have to
be written in the high-level source language, which is often not realistic.
Compiler correctness would be sufficient on its own if all components
could be written in a safe high-level language.
The point in moving from compiler correctness to full abstraction,
which is stronger, is precisely to account for the fact that some
components have to be written in the low-level language, \EG for
performance reasons.

Assumption (2) breaks as soon as we consider that it makes a
difference whether the attacker owns one or all the untrusted components.
As an example, assume that an attacker succeeds in taking over an
untrusted component that was used by the program to render the picture
of a cat.
Would one care whether this allows the attacker to also take over the
low-level cryptographic library that manages private keys?
We believe that the cryptographic library, which is a
security-critical component, should get the same level of protection
as a compiled component, even if for efficiency it is implemented in
the low-level language.

When assumption (1) breaks, trusted components need to be protected
from each other, or at least from the potentially memory unsafe ones.
When assumption (2) breaks, untrusted security-critical components
need to be protected from the other untrusted components.
In this work, we propose a stronger attacker model that removes both these
assumptions by requiring all components to be protected from each other.
%

\subsection{Mutual Distrust Attacker Model}
\label{sec:mutual-distrust}

\bcp{I'm still not very clear on whether this attacker model is an
  instance of classical full abstraction, or a minor variant, or a
  significant generalization.  To figure this out, it would really help to
  see the property you are describing written out rigorously /
  mathematically.} 

\begin{figure}
\includegraphics[trim=1.2cm 21.7cm 8cm
1cm,clip,width=\columnwidth]{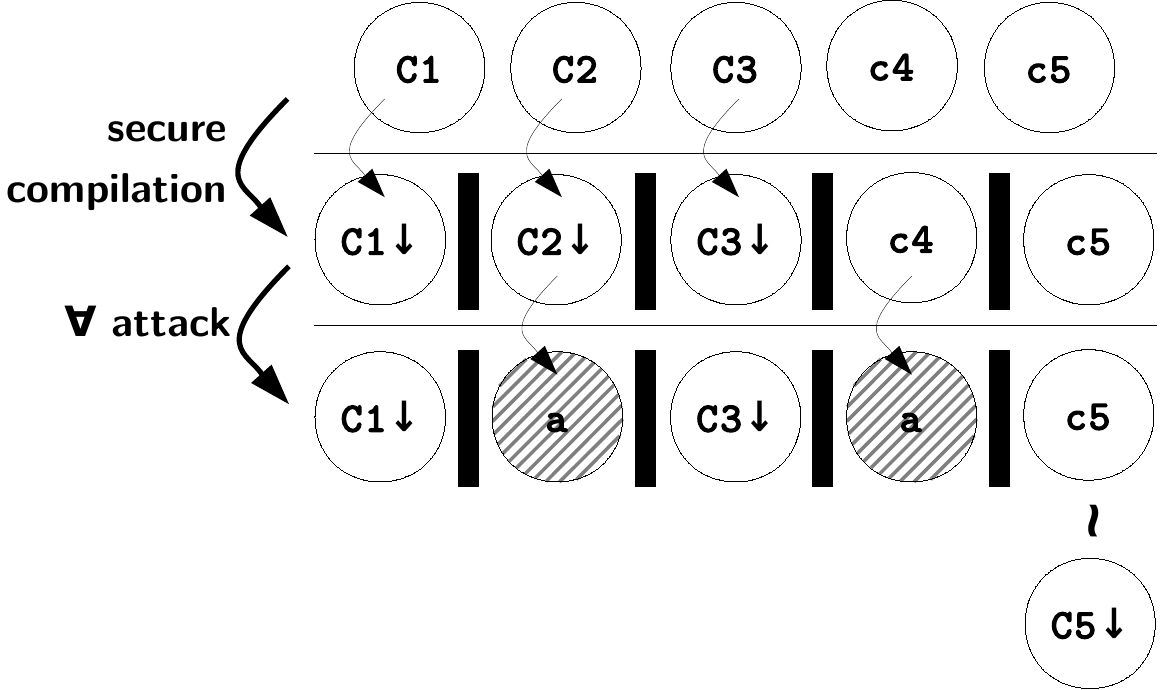}
\caption{Secure compilation for mutually distrustful
components}
\label{fig:mutual-distrust}
\end{figure}

We propose a new attacker model that overcomes the previously
highlighted limitations of full abstraction.
In this attacker model, we assume that each component could be
compromised and protect all the other components from it:
we call it an attacker model for \emph{mutually distrustful components}.
This model can provide security even in C-like unsafe languages when
some of the high-level components are memory unsafe or have undefined
behaviors. This is possible if the high-level semantics treats
undefined behavior as arbitrary changes in the state of the component
that triggered it, rather than in the global state of the program.
In the following we will assume the high-level language is secure.

All compiled high-level components get security unconditionally:
the secure compiler and the dynamic barriers protect them
from all low-level attacks, which allows reasoning about their
security in the high-level language.
For low-level components to get security they have to satisfy
additional conditions, since the protection barriers are often not
enough on their own for security, as the compiler might be inserting
boundary checks, cleaning registers, \ETC and the low-level code still
needs to do these things on its own in order to get full protection.
Slightly more formally, in order for a low-level component {\bf c} to
get security it must behave in all low-level contexts like some
compiled high-level component {\bf C$\downarrow$}.
In this case we can reason about its security at the high level by
modelling it as {\bf C}.
This captures the scenario in which {\bf c} is written in the
low-level language for efficiency reasons.

We illustrate our stronger attacker model in
figure~\ref{fig:mutual-distrust}.
The protected program is the same as in the previous full abstraction
diagram of figure~\ref{fig:multi-component-full-abstraction}.
This time, however, the user doesn't choose a trust barrier:
all components are considered mutually distrustful instead.
Each of them gets protected from the others thanks to barriers
inserted by the compiler.
While components {\bf C3}, {\bf c4}, and {\bf c5} were distrusted and
thus not protected in the previous diagram, here all of them can get
the same amount of protection as other components.
To get security {\bf C3} is compiled using the secure compiler, while
for {\bf c4} and {\bf c5} security is conditioned on equivalence
to high-level components; in the figure we assume this only for {\bf c5}.
The attacker can compromise arbitrary components (including
high-level compiled components), \EG {\bf C2$\downarrow$} and {\bf c4}
in the diagram.
In this compromise scenario, we ensure that the uncompromised components {\bf
  C1$\downarrow$}, {\bf C3$\downarrow$}, and {\bf c5} are protected
from all low-level attacks coming from the compromised components.
In general, our attacker model defends against all such compromise
scenarios.

To sum up, our attacker model can be stated as follows:
(a) the attacker compromises with component granularity,
(b) the attacker may compromise any set of components,
(c) in every compromise scenario, each uncompromised compiled
high-level component is secure against low-level attacks from all
compromised components, and
(d) in every compromise scenario, each uncompromised low-level
component that has a high-level equivalent is secure against low-level
attacks from all compromised components.




\section{Micro-Policies and the PUMP: Efficient Tag-Based Security
Monitors}

\ch{move this as late as possible (putting related
  work so early on is a mistake, see Simon PJ's advice on writing good papers).}
\ch{moving this later probably means splitting it up}

We present \emph{micro-policies}~\cite{pump_asplos2015,
micropolicies2015}, the mechanism we use to
monitor low-level code so as to enforce that our compiler is secure.
Micro-policies~\cite{pump_asplos2015,micropolicies2015}
are a tag-based dynamic protection mechanism for machine code.
The reference implementation on which micro-policies are based is
called the PUMP~\cite{pump_asplos2015} (Programmable Unit for Metadata
Processing).

The PUMP architecture associates each piece of data in the system with
a {\em metadata tag} describing its provenance or purpose (\EG
``this is an instruction,'' ``this came from the network,'' ``this is
secret,'' ``this is sealed with key $k$''), propagates this metadata
as instructions are executed, and checks that policy rules are obeyed
throughout the computation.
It provides great flexibility for defining policies and puts no
arbitrary limitations on the size of the metadata or the number of
policies supported.
Hardware simulations show~\cite{pump_asplos2015} that an Alpha
processor extended with PUMP hardware achieves performance comparable
to dedicated hardware on a standard benchmark suite when enforcing
either memory safety, control-flow integrity, taint tracking, or code
and data separation.
\def\rtover{10}%
\def\eover{60}%
\def\pceil{10}%
\def\aover{110}%
When enforcing these four policies simultaneously, monitoring imposes
modest impact on runtime (typically under \rtover\%) and power ceiling
(less than \pceil\%), in return for some increase in energy usage
(typically under \eover\%) and chip area (\aover\%).

The reference paper on micro-policies~\cite{micropolicies2015}
generalizes previously used methodology~\cite{PicoCoq2013} to provide
a generic framework for formalizing and verifying \emph{arbitrary}
policies enforceable by the PUMP architecture.
In particular, it defines a \emph{generic symbolic machine}, which
abstracts away from low-level hardware details and serves
as an intermediate step in correctness proofs.
This machine is parameterized by a {\em symbolic micro-policy},
provided by the micro-policy designer, that expresses tag propagation
and checking in terms of structured mathematical objects rather than
bit-level concrete representations.
The micro-policies paper also defines a \emph{concrete machine}
which is a model of PUMP-like hardware, this time including implementation
details.

The proposed approach to micro-policy design and verification is
presented as follows.
First, one designs a reference \emph{abstract machine}, which will
serve as a micro-policy specification.
Then, one instantiates the generic symbolic machine with a
symbolic micro-policy and proves that the resulting symbolic
machine \emph{refines}
the abstract machine: the observable behaviors of the symbolic
machine are also legal behaviors of the abstract machine, and in
particular the symbolic machine fail-stops whenever the abstract
machine does.
Finally, the symbolic micro-policy is implemented in low-level terms,
and one proves that the concrete machine running the micro-policy
implementation refines the symbolic machine.

In this work, we use a slightly modified symbolic machine as
the target of our secure compiler.
Our symbolic machine differs from the previous
one~\cite{micropolicies2015} in two ways:
First, its memory is separated into regions, which are addressed by
\emph{symbolic pointers}.
%
Note, however, that our protection does not rely on this separation
but only on the tagging mechanism itself: in particular, all
components can refer to all symbolic pointers, without restriction.
Mapping memory regions to concrete memory locations before executing
the program on a concrete machine would be a main task of
a \emph{loader} --- we leave a complete formalization of loading for
future work.
Second, we extend the micro-policy mechanism
itself, allowing rules to read and write the tags on more components of
the machine state.
We detail these extensions in section \ref{sec:target}, which is
dedicated to our target machine.
%
We also leave for future work the implementation of these additional
tags in the PUMP rules, their formalization in an extended
concrete machine, and the extension of our results to
the concrete level.


\section{Compilation Chain Overview}
\label{sec:overview}




\iflater
\else
In this section, we present our compiler and give intuition about the
connections between the different parts that play a role in our
solution.
\fi

\yj{Idea: low-level linking = put two partial programs together after
                              checking that their interfaces are
                              compatible;
          low-level loading = tag a complete program to enforce that
                              all components will correctly implement
                              their interfaces at runtime
          \ch{(Old) Really? It does not type-check them!}\yj{Updated.}}\bcp{These are not the standard uses of the
          words...}\yj{I thought that we matched with standard uses
          except that we probably do ``more'' than usual
          linking/loading? What would be the standard uses then?}
\bcp{My usage (which I thought was fairly standard) is: linking =
          combining components, resolving cross-references, and checking
          completeness; loading = relocation and other low-level tasks
          related to creating an executable memory image }

\begin{figure}
\centering
\includegraphics[trim=1.9cm 14.4cm 1.3cm 2.3cm,clip,width=\columnwidth]{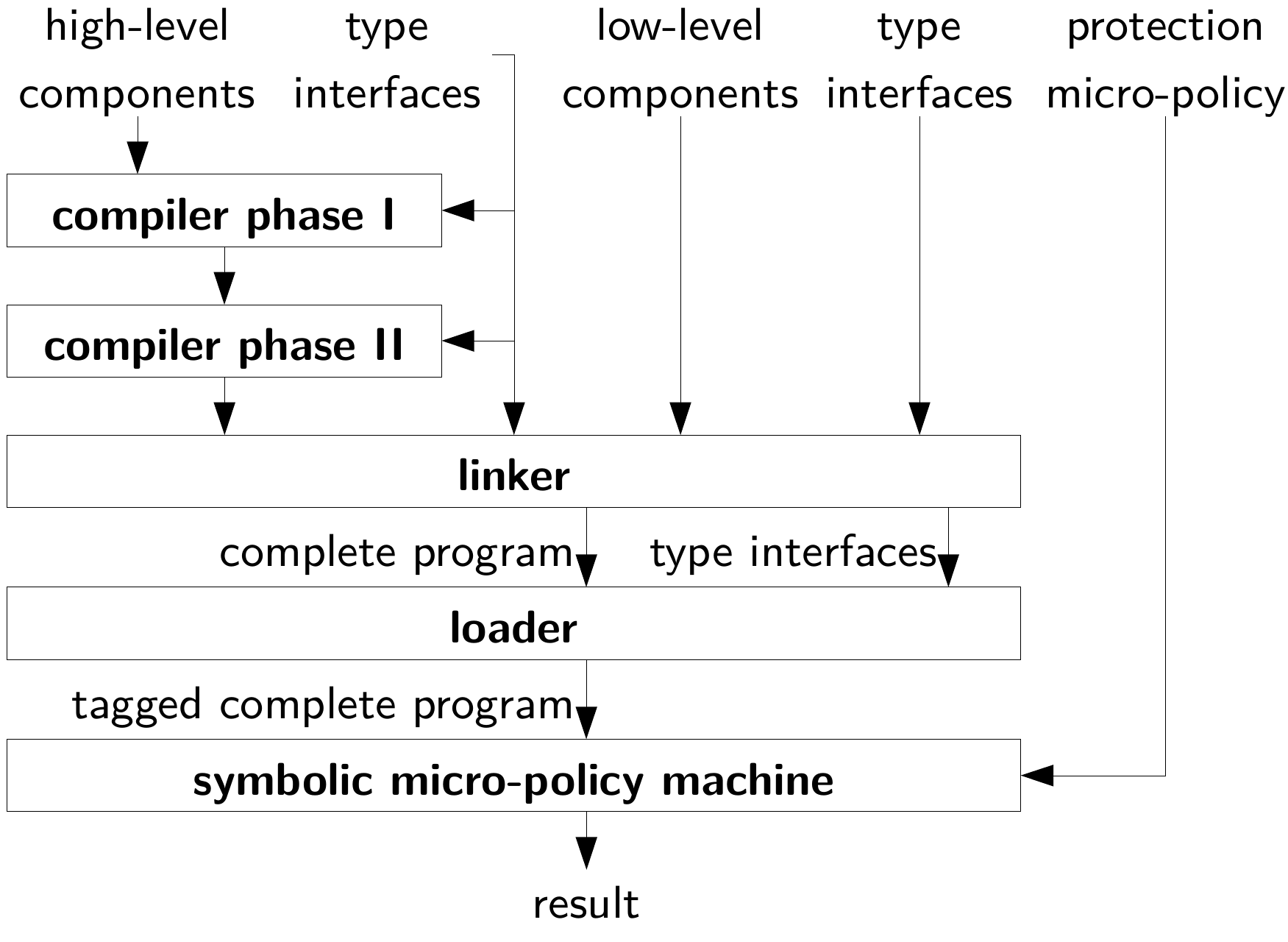}
\caption{Overview of the compilation chain}

\label{fig:chain}
\end{figure}

Our compilation chain, which we present in figure~\ref{fig:chain},
splits into a two-step compiler, a linker and a loader.
It produces a program to execute on the symbolic micro-policy machine.
Our dedicated protection micro-policy will be loaded into the machine,
allowing proper runtime monitoring of the program.

The compilation chain takes source components (\EG a main program and
standard libraries) and target components (\EG low-level libraries) as
input, and outputs a target-executable program.
Components must all come with interface specifications, written in
a common interface language.
These interfaces specify the definitions that each component provides,
and the definitions it expects other components to provide.

In the compilation phase, the compiler first translates source
components to an intermediate representation, which then gets
translated to target components.

In the linking phase, the linker checks that the interfaces of the
components getting linked are compatible.
It also makes sure that all components only give definitions under
names that follow from their interfaces; and symmetrically that they
do provide a definition for each of these names.
Is so, the linker puts them together to form a partial program, and
makes sure that this partial program is actually complete (\IE no
definition is missing).

In the loading phase, the loader builds an initial machine state out
of this complete program by tagging its memory using type information
that was gathered by the linker.
The result is thus a target-level executable program --- \IE a
symbolic machine tagged program.
The loader's tagging will force all components to correctly implement
the interfaces that was provided for them, when we later run and
monitor them using our protection micro-policy:
The machine will failstop as soon as a component violates its
interface upon interaction with another component (violations on
internal computational steps is not harmful and hence allowed).

Because we required that every component should have an interface,
low-level libraries that were compiled using other compilers --- \EG C
compilers --- are only usable in our system if somebody writes
interfaces for them.
In section~\ref{sec:discussion} we discuss more generally the need for
manual or automated \emph{wrapping} of low-level code; once we have a
way to wrap low-level code, providing a matching interface will be
easy.


\section{Languages and Machines}
\label{sec:languages}

In this section, we present and formalize the languages that are used
in the compilation scheme.
We introduce a simple object-oriented language, an abstract stack
machine, and an extended symbolic micro-policy machine with segmented
memory.
The first will be our source language and the last our target machine,
while the intermediate abstract machine will offer a different view
from the source language which makes the connection with the low level
more direct.
The source language includes constructs for specifying the interfaces
of components; these get reused as-is at all levels.

\subsection{Source Level: An Object-Oriented Language}
\label{sec:source}

We first formalize our source language,
beginning with the interface constructs that it offers
and then presenting its syntax and semantics.
The source language we consider is an imperative class-based
object-oriented language with static objects, private fields, and
public methods.
It is inspired by previous formalizations of Java core
calculi~\cite{DBLP:journals/toplas/IgarashiPW01,
DBLP:journals/entcs/BiermanP03}
and Java subsets~\cite{DBLP:conf/esop/JeffreyR05}, with the aim of
keeping things as simple as possible.
As a consequence, we do not handle inheritance nor dynamic allocation,
which are available in all these works.

We start with the simplicity of Featherweight
Java~\cite{DBLP:journals/toplas/IgarashiPW01}, and add imperative
features in the way Middleweight
Java~\cite{DBLP:journals/entcs/BiermanP03} does.
However, we do not add as many imperative features: just branching,
field update and early termination (the latter is not a feature of
Middleweight Java).
The resulting language is similar to Java
Jr.~\cite{DBLP:conf/esop/JeffreyR05} with early termination, but
without packages:
Our granularity for components is that of classes instead.

Example components which encode some usual types are provided in
appendix section~\ref{sec:encoding}, and could help in getting
familiar with this language.

\subsubsection{Interfacing: A Specification Language for Communicating
Classes}

\renewcommand{\ottgrammartabular}[1]{\medskip\par{\raggedright #1}\par\medskip}
\renewcommand{\ottrulehead}[3]{{$#1$  $#2$ \hfill {#3}} \\ }
\renewcommand{\ottprodline}[6]{{}\ ${}#1$\mbox{\ {}$#2$}{}}
\renewcommand{\ottfirstprodline}[6]{{}\ ${}$\mbox{\ {}$#2$}{}}
\renewcommand{\ottprodnewline}{}
\renewcommand{\ottinterrule}{\\[1em]}
\renewcommand{\ottnt}[1]{ { \mathit{#1} } }
\renewcommand{\ottmv}[1]{ { \mathit{#1} } }

\begin{figure}
\ottgrammartabular{
\ottIDT\ottinterrule
\ottEDT\ottinterrule
\ottDT\ottinterrule
\ottCDT\ottinterrule
\ottCD\ottinterrule
\ottMD\ottinterrule
\ottODT\ottinterrule
\ottOD\ottafterlastrule
}
\caption{Interface language syntax\bcp{The declarations are in a kind of
    random order.  Are $c$ and $o$ identifiers?}\yj{Put the text on
    $\mathit{c}$ and $\mathit{o}$ earlier in the accompanying text. As for the
    order of declarations I'm not sure which one would be less random
    than the current one, which make senses for me.}}
\label{fig:interfacingsyntax}
\end{figure}

The notion of component in this work is that of a class $\mathit{c}$
together with all its object instances' definitions.
%
Because we have no dynamic allocation, for the moment these instances
are simply all the static objects defined with type $\mathit{c}$.
To allow interaction between components while being able to separately
compile them, we have a simple interface syntax based on import and
export declarations.
This interface language gives the external view on source components
and is presented in figure~\ref{fig:interfacingsyntax}.

\paragraph{Syntax and Naming Conventions}

Object names $\mathit{o}$ and class names $\mathit{c}$ are global and assumed to
be arbitrary natural numbers.
%
%
The two main syntactic constructs in the interface language are
class declarations and static object declarations.
Class declarations specify public methods with argument and result
types without providing their body; no fields are declared in an
interface because we only consider private fields in our languages.
Static object declarations specify an object to be an instance of a
given class, without providing the values of its fields.

The interface of a partial program at all levels is composed of an
import declaration table $\ottnt{IDT}$ specifying the class and object
definitions it expects other components to provide, and an export
declaration table $\ottnt{EDT}$ which declares the definitions that this
partial program provides.
Export and import declaration tables share common syntax and are
composed of class and object declarations.
The type of the declared objects must come from the classes specified
by the partial program: defining object instances of classes coming
from other components is not allowed.
Intuitively, our object constructors (and fields) are private and the
only way to interact with objects is via method calls.

In contrast with objects and classes to which we refer using global
names, methods are class-wise indexed: the methods $\ottmv{m}$ of a class
$\mathit{c}$ are referred to as $1, \dots, k$ following the order of their
declarations.
(The same goes for fields $\ottmv{f}$, below.)
The syntax we consider for names can be thought of as coming out of a
parser, that would take a user-friendly Java-like syntax and perform
simple transformations so that the names match our conventions.

\paragraph{Use at Linking and Loading}


We have presented the roles of the linker and the loader when we
introduced the compilation chain in section~\ref{sec:overview}:
%
%
\iflater
\else
We can model linking as an operation that takes two target partial
programs and their interfaces, and yields a new partial program which
contains all definitions from both partial programs, with a new
matching interface.
Loading then takes a complete target program and tags it, yielding a
machine state that can be reduced using the semantics of our symbolic
micro-policy machine.
\fi
Let us now explain how interfaces are used at linking and loading.

A class import declaration gives the precise type signatures that the
partial program expects from the methods of the corresponding class:
When linking against a partial program that defines this class, the
class export declaration should exactly match with the import one.
Similarly, an import object declaration gives the expected type for
this object, and should match with the corresponding export
declaration when linking against a partial program that defines
it.

Two partial programs have compatible interfaces when (1) they
don't both have export declarations for the same class nor the same
object, and (2) for every name in an import declaration of any of the
two programs, if the other program happens to have an export
declaration for this name, then the import and export declarations are
syntactically equal.
Linking two partial programs with compatible
interfaces yields a new partial program with updated import/export
declarations:
Export declarations are combined, and import declarations that
found matching export declarations are removed.
When all partial programs have been linked together, the linker
can check that there are no remaining import declarations to make sure
that the program is complete.

Finally, the loader will make use of the export declarations to ensure
that programs comply with the export declarations they declared:
In the untyped target language, the loader sets the initial memory tags
in accordance with the export declarations, which will allow our
micro-policy to perform dynamic type checking.
%
This will be further explained in section \ref{sec:typesafety}.
%

\subsubsection{Source Syntax and Semantics}


\begin{figure}
\ottgrammartabular{
\ottSP\ottinterrule
\ottT\ottinterrule
\ottCT\ottinterrule
\ottC\ottinterrule
\ottM\ottinterrule
\otte\ottinterrule
\ottOT\ottinterrule
\ottO\ottafterlastrule
}
\caption{Source language syntax}
\label{fig:sourcesyntax}
\end{figure}

The syntax of our source language is presented in
figure \ref{fig:sourcesyntax}.
The two main syntactic constructs in this language are class
definitions and static object definitions.
Class definitions declare typed private fields and define public
methods with argument and result types as well as an expression which
serves as a method body.
Static object definitions define instances of defined classes by
providing well-typed values for the fields.
For simplicity, methods only take one argument: this does not affect
expressiveness because our language is expressive enough to encode
tuple types
(appendix section~\ref{sec:encoding} shows examples that encode
more complex types than tuple types).

Most expressions are not surprising for an object-oriented language:
apart from static object references $\mathit{o}$ and variables ($\ottkw{this}$
for the current object and $\ottkw{arg}$ for the method's argument), we
have support for selecting private fields, calling public methods, and
testing object identities for equality.
The language also features field update and early termination.
Both are crucial for modeling realistic low-level attackers in our
high-level language:
Low-level attackers can indeed keep information between calls using
the memory and stop the machine whenever they have control.
We thus add primitives that can model this in the high-level:
field updates enable methods to have state (methods are not pure
functions anymore), and early termination allows an attacker to
prematurely end the program.

Like we already mentioned, fields are private and methods are public.
This means that in the method body of a given class, the only objects
whose fields may be accessed are the instances of that specific class.
The only way to interact with object instances of other classes is to
perform a method call.

The only values in the language are object names, and the only types
are class names.
The language comes with a type system that ensures that object
and method definitions match with the types that were
declared for them.
Our language does not feature dynamic allocation, inheritance, or
exceptions.
We hope to add some of these features in the future.
Loops are simulated using recursive method calls and
branching is done via object identity tests.

The semantics of the source language is standard and is given
in appendix~\ref{sec:sourcesem}.





\subsection{Intermediate Level: An Object-Oriented Stack Machine}

\begin{figure}
\ottgrammartabular{
\ottIP\ottinterrule
\ottICT\ottinterrule
\ottIC\ottinterrule
\ottIM\ottinterrule
\ottIcode\ottinterrule
\ottIinstr\ottinterrule
\ottLOT\ottinterrule
\ottLO\ottinterrule
\ottLS\ottafterlastrule
}
\caption{Intermediate language syntax}
\label{fig:intermsyntax}
\end{figure}

Our intermediate machine is an object-oriented stack machine with one
local stack per class.
The syntax for intermediate machine programs is presented in
figure~\ref{fig:intermsyntax}.
%
The main syntactic construct is that of a compartment, which is the
notion of component at this level.
A compartment combines a class definition with all the object
instances of this class and with a local stack.

The main difference with respect to the source language is that instead of
expressions, method bodies are expressed as sequences of instructions
in abstract machine code.
These instructions manipulate values stored on the local stack
associated with each compartment.

$\ottkw{Nop}$ does nothing.
$\ottkw{This}$, $\ottkw{Arg}$ and $\ottkw{Ref} \, \mathit{o}$ put an object on the local stack
--- the current object for $\ottkw{This}$, the current method argument for
$\ottkw{Arg}$, and
object $\mathit{o}$ for $\ottkw{Ref} \, \mathit{o}$.
$\ottkw{Sel} \, \ottmv{f}$ pops an object from the stack, selects field $\ottmv{f}$ of
this object and pushes the selected value back to the stack.
$\ottkw{Upd} \, \ottmv{f}$ pops a value and an object, sets the $\ottmv{f}$ field of
the object to the popped value, then pushes back this value on the
stack.
$\ottkw{Call} \, \mathit{c} \, \ottmv{m}$ pops an argument value and a target object $\mathit{o}$,
performs a call of the object $\mathit{o}$'s $\ottmv{m}$ method with the popped
argument if $\mathit{o}$ has type $\mathit{c}$ (otherwise, the machine
failstops).
%
The callee can then use the $\ottkw{Ret}$ instruction to give control
back to the caller: this instruction pops a value from the callee's
stack and pushes it on the caller's stack, as a result value
for the call.
$\ottkw{Skip} \, \mathit{n}$ skips the $\mathit{n}$ next instructions.
%
$\ottkw{Skeq} \, \mathit{n}$ pops two values from the local stack and either skips the
$\mathit{n}$ next instructions if the values are equal, or does nothing
more if they are different.
$\ottkw{Drop}$ removes the top value from the stack.
$\ottkw{Halt}$ halts the machine immediately, the result of the
computation being the current top value on the stack.

The purpose of this intermediate language is to ease going
low-level.
In particular, its syntax with the $\ottkw{Call}$ instruction being
annotated with a class makes explicit the fact that method calls are
statically resolved by the source to intermediate compiler.
This is possible in our source language, because we have no
inheritance.

\subsection{Target Level: An Extended Symbolic Micro-Policy Machine}
\label{sec:target}


Here, we present our the target of our compiler: an extended symbolic
micro-policy machine with segmented memory.
We first recall the common basis for our machine and the symbolic
machine presented in the original micro-policies paper~\cite{micropolicies2015},
then present and comment on the differences between them.

\subsubsection{Symbolic Micro-Policy Machine}

\begin{figure}
\ottgrammartabular{
\ottmem\ottinterrule
\ottloc\ottinterrule
\ottR\ottinterrule
\ottword\ottinterrule
\ottinstr\ottinterrule
\ottimm\ottafterlastrule
}
\caption{Symbolic machine memory}
\label{fig:symbmem}
\end{figure}

\begin{figure}
\ottgrammartabular{
\ottLP\ottinterrule
\ottLPmem\ottinterrule
\ottLPR\ottinterrule
\ottimm\ottafterlastrule
}
\caption{Symbolic machine program syntax}
\label{fig:symbprog}
\end{figure}

A symbolic micro-policy machine~\cite{micropolicies2015} is an
executable model of micro-policies that abstracts away from some of
the implementation details (e.g. the implementation of the
micro-policy monitor in machine code).
The definition of a symbolic micro-policy machine is abstracted over
a \emph{symbolic micro-policy}.

In our case, a symbolic micro-policy is defined as a collection
of \emph{symbolic tags}, which are used to label instructions and
data, and a \emph{transfer function}, which is invoked on each step of
the machine to propagate tags from the current machine state to the
next one.
We ignore \emph{monitor services} of Azevedo de Amorim et
al.~\cite{micropolicies2015} and \emph{extra pieces of state}
which are only used by monitor services --- because we don't need
them: we successfully avoid monitor services, which in the context of
micro-policies are much more expensive than rules.
The transfer function is a mathematical function that defines the
micro-policy rules --- in the mechanized metatheory of the original
micro-policies paper this function is written in Gallina, the
purely-functional language at the heart of the Coq proof assistant.

A machine state of our symbolic micro-policy machine is composed of a
memory, a register file of general-purpose registers, and a program
counter register $\ottnt{pc}$.
The program counter register points to a location in memory which
contains the next instruction to execute.

We present the list of instructions in figure~\ref{fig:symbmem}, together with
other memory-related definitions on which we will focus when we
later explain the segmented memory.
These instructions are those of the machine code of the original
micro-policies paper~\cite{micropolicies2015}:
$\ottkw{Nop}$ does nothing.
$\ottkw{Const} \, i \, {r_{\textsf{d} } }$ puts an immediate constant $i$ into register
${r_{\textsf{d} } }$.
$\ottkw{Mov} \, {r_{\textsf{s} } } \, {r_{\textsf{d} } }$ copies the contents of ${r_{\textsf{s} } }$ into
${r_{\textsf{d} } }$.
$ { \ottkw{Binop} _{ \ottnt{op} } }  \mathit{r_{{\mathrm{1}}}}   \mathit{r_{{\mathrm{2}}}}   {r_{\textsf{d} } } $ performs a binary operation $\ottnt{op}$ (\EG
addition, subtraction, equality test) on the content of registers
$\mathit{r_{{\mathrm{1}}}}$ and $\mathit{r_{{\mathrm{2}}}}$, and puts the result in register ${r_{\textsf{d} } }$.
$\ottkw{Load} \, {r_{\textsf{p} } } \, {r_{\textsf{s} } }$ copies the content of the memory cell at the
memory location stored in ${r_{\textsf{p} } }$ to ${r_{\textsf{s} } }$.
$\ottkw{Store} \, {r_{\textsf{p} } } \, {r_{\textsf{s} } }$ copies the content of ${r_{\textsf{s} } }$ to the memory
cell at the memory location stored in ${r_{\textsf{p} } }$.
$\ottkw{Jump}$ and $\ottkw{Jal}$ (jump-and-link) are unconditional
indirect jumps, while $\ottkw{Bnz} \, \mathit{r} \, i$ branches to a
fixed offset $imm$ (relative to the current pc) if register $\mathit{r}$ is
nonzero.
$\ottkw{Halt}$ halts the machine immediately.

In the following, we denote $ { \ottkw{Binop} _{ \ottsym{+} } }  \mathit{r_{{\mathrm{1}}}}   \mathit{r_{{\mathrm{2}}}}   {r_{\textsf{d} } } $ (addition) by
$\ottkw{Add} \, \mathit{r_{{\mathrm{1}}}} \, \mathit{r_{{\mathrm{2}}}} \, {r_{\textsf{d} } }$, and $ { \ottkw{Binop} _{ \ottsym{-} } }  \mathit{r_{{\mathrm{1}}}}   \mathit{r_{{\mathrm{2}}}}   {r_{\textsf{d} } } $ (subtraction) by
$\ottkw{Sub} \, \mathit{r_{{\mathrm{1}}}} \, \mathit{r_{{\mathrm{2}}}} \, {r_{\textsf{d} } }$.

\subsubsection{Extensions to Monitoring Mechanism}

We consider a more powerful symbolic micro-policy machine that allows
the transfer function to inspect and update more tags.

First, we assume that the transfer function produces new tags for the
input arguments of the instructions; not only for the output one.
This is required, for instance, to transfer a linear capability from an input
to an output: one wants not only to copy the capability in the output tag,
but also to erase it from the input tag.

Second, we assume that there are some fixed registers whose
tags can always be checked and updated by the transfer function, even
if the registers are neither input nor output to the current
instruction.
This allows us to efficiently clean these fixed registers upon
\chrev{calls and returns}.

Third, we assume that the transfer function receives as an argument
not only the tag of the current instruction, but also the tag on the
{\em next} instruction.
For instance, when monitoring a $\ottkw{Jump}$ instruction, we assume the
tag on the targeted location can be checked.
This extension allows us to write and explain our micro-policy in a
much simpler way.

The first two extensions require extra hardware support.
%
%
For the last extension, however, we conjecture that our micro-policy
--- and probably other similar micro-policies --- can be
transformed into a policy which doesn't have to perform next
instruction checks.
This kind of translation has already been done by hand,
for example in a previous compartmentalization
micro-policy~\cite{micropolicies2015}: the next
instruction checks are replaced by current instruction checks
happening on the next step, making the machine failstop one step later
in the case of an error.
We leave for future work the writing of a policy compiler
doing such a transformation automatically.

\subsubsection{Easing Component-Oriented Reasoning: Segmented Memory}


Instead of having a monolithic word-addressed memory, the machine we
consider has a segmented memory which consists of several memory
regions which are addressed by \emph{symbolic locations}.
Targeting such a machine allows for easy separate compilation, and
is a pragmatic intermediate step when going towards real machine code,
which we plan to do in the future.

As presented in figure~\ref{fig:symbmem}, the definition of memory
directly mentions symbolic locations.
A generic symbolic machine definition would be abstracted over the
definition of symbolic locations, but
in our case, we define them to be either method, object, or stack
locations for reasons that will be clear when we present our
compiler in section~\ref{sec:i2t}.
Our instantiation of memory tags ${t_{\textsf{mem} } }$ will be studied with other
definitions related to the symbolic micro-policy, in
section~\ref{sec:micro-policy}.

Immediate constants and words in the symbolic machine are extended to
include symbolic locations with an offset, which are memory addresses:
The $\ottmv{k}$ memory cells of a region located at $\ottnt{loc}$ are
addressed from $\ottnt{loc}  \ottsym{+}  \ottsym{0}$ to $\ottnt{loc}  \ottsym{+}  \ottsym{(}  \ottmv{k}  \ottsym{-}  \ottsym{1}  \ottsym{)}$.
In the following, we use the simpler notation
$\ottnt{loc}$ for $\ottnt{loc}  \ottsym{+}  \ottsym{0}$.

Words are also extended to include a new $\ottkw{encode} \, \ottnt{instr}$ construct:
Decoding instructions in this machine is a symbolic operation of
deconstructing such a word.
Now that instructions feature symbolic locations with an offset as
immediate values, it would indeed have no practical meaning to try to
extend the encoding of instructions to this.
When we use the PUMP in future work, some of the symbolic-level
instructions could have to be compiled to sequences of PUMP
instructions:
PUMP memory addresses, the PUMP equivalent of symbolic locations with
an offset, are word-sized and thus too big to fit in a PUMP immediate
value.
Another solution would be to restrict addressable memory to less than
a full word; the symbolic encoding allows us to delay the choice
between these solutions.

Tags are not affected by all the changes that we listed, hence the
monitoring mechanism isn't affected either.
The semantics, however, is affected in the following way:
Trying to decode an instruction out of a regular word or of a symbolic
location with an offset failstops the machine.
All instructions failstop when one of their operands has the form
$\ottkw{encode} \, \ottnt{instr}$.
A $\ottkw{Jump}$, $\ottkw{Jal}$, $\ottkw{Load}$ or $\ottkw{Store}$ instruction used
with a regular word for the pointer failstops the machine.
These instructions keep their usual behavior when provided with a
symbolic location and a valid offset; if it does not correspond to a
valid memory address, however, the machine failstops.
Most binary operations failstop when the first or second
operand is a symbolic location with an offset:
exceptions are
1) addition and subtraction with regular words, when the first operand
is the location with an offset, which respectively increment and
decrement the offset accordingly;
2) equality tests between symbolic locations with offsets.
Finally, the $\ottkw{Bnz}$ instruction failstops when the
provided register or immediate value is a symbolic location with an
offset.


The syntax for symbolic machine programs is presented in figure~\ref{fig:symbprog}.
They define a memory which is like the one of the symbolic machine,
except that cells are not tagged:
The tags for the memory will be provided by the low-level loader,
which will be detailed in the next section.


\section{Our Solution: Protecting Compiled Code with a Micro-Policy}
\label{sec:system}

In this section, we present our solution for the secure compilation of
mutually distrustful components: 
first we describe our simple two-step compiler, then we present
our micro-policy dynamically protecting components from each other.

\subsection{Two-Step Compilation}
\label{sec:compiler}

We start with our two-step compilation: first the compilation of
source programs to intermediate machine programs, then the one of
intermediate machine programs to target machine programs.

\subsubsection{From Source to Intermediate}


Our type-preserving source to intermediate compiler is a mostly direct
translation of source expressions to abstract machine instructions,
which gives a lower-level view on source components.
Nothing in the translation is related to security: we provide security
at this level by giving appropriate semantics to intermediate-level
instructions, which make them manipulate local stacks and local object
tables rather than a single global stack and a single object table.
In this translation, we statically resolves method calls, which is
possible because our language doesn't feature inheritance.

A high-level component is easily mapped to an intermediate
compartment:
Method bodies are compiled one by one to become intermediate-level
method bodies.
Object definitions corresponding to that component are taken from
the global object table $\ottnt{OT}$ and put in the local object table
$\ottnt{LOT}$ of the compartment.
Finally, the stack is initialized to an empty stack.

\paragraph{Compiling Source Expressions to Stack Machine Code}

\begin{figure}
\begin{align*}
 \mathscr{A} ( \ottkw{this} )  =\;& \ottkw{This} \\
 \mathscr{A} ( \ottkw{arg} )  =\;& \ottkw{Arg} \\
 \mathscr{A} ( \mathit{o} )  =\;& \ottkw{Ref} \, \mathit{o} \\
 \mathscr{A} ( \ottnt{e}  \ottsym{.}  \ottmv{f} )  =\;&
    \mathscr{A} ( \ottnt{e} )   ;  \ottkw{Sel} \, \ottmv{f}  \\
 \mathscr{A} ( \ottnt{e}  \ottsym{.}  \ottmv{f}  \ottsym{:=}  \ottnt{e'} )  =\;&
     \mathscr{A} ( \ottnt{e} )   ;   \mathscr{A} ( \ottnt{e'} )    ;  \ottkw{Upd} \, \ottmv{f}  \\
 \mathscr{A} ( \ottnt{e}  \ottsym{.}  \ottmv{m'}  \ottsym{(}  \ottnt{e'}  \ottsym{)} )  =\;&
     \mathscr{A} ( \ottnt{e} )   ;   \mathscr{A} ( \ottnt{e'} )    ;  \ottkw{Call} \, \mathit{c'} \, \ottmv{m'}  \\
  & \mbox{where } c' \mbox{ satisfying} \\
  & \mbox{~~} \ottnt{P}  \ottsym{;}  \mathit{c}  \ottsym{,}  \ottmv{m}  \vdash  \ottnt{e}  \ottsym{:}  \mathit{c'} \\
  & \mbox{is found by type inference} \\
 \mathscr{A} ( \ottnt{e_{{\mathrm{1}}}}  \ottsym{==}  \ottnt{e_{{\mathrm{2}}}}  \;?\;  \ottnt{e_{{\mathrm{3}}}}  \ottsym{:}  \ottnt{e_{{\mathrm{4}}}} )  =\;&
    \mathscr{A} ( \ottnt{e_{{\mathrm{1}}}} )   ;   \mathscr{A} ( \ottnt{e_{{\mathrm{2}}}} )   ; \\
  & \ottkw{Skeq} \, \ottsym{(}  \ottsym{\mbox{$\mid$}}   \mathscr{A} ( \ottnt{e_{{\mathrm{4}}}} )   \ottsym{\mbox{$\mid$}}  \ottsym{+}  \ottsym{1}  \ottsym{)} ; \\
  &   \mathscr{A} ( \ottnt{e_{{\mathrm{4}}}} )   ;  \ottkw{Skip} \, \ottsym{\mbox{$\mid$}}   \mathscr{A} ( \ottnt{e_{{\mathrm{3}}}} )   \ottsym{\mbox{$\mid$}}  ; \\
  &  \mathscr{A} ( \ottnt{e_{{\mathrm{3}}}} )  ; \\
  & \ottkw{Nop} \\
 \mathscr{A} ( \ottnt{e}  \ottsym{;}  \ottnt{e'} )  =\;&
     \mathscr{A} ( \ottnt{e} )   ;  \ottkw{Drop}   ;   \mathscr{A} ( \ottnt{e'} )   \\
 \mathscr{A} ( \ottkw{exit} \, \ottnt{e} )  =\;&
    \mathscr{A} ( \ottnt{e} )   ;  \ottkw{Halt} 
\end{align*}
\caption{Compiling source expressions to
intermediate machine code}
\label{fig:s2i}
\end{figure}

Assuming that method $\ottmv{m}$ of class $\mathit{c}$ in program $\ottnt{P}$
has definition ${c_{\textsf{r} } }  \ottsym{(}  {c_{\textsf{a} } }  \ottsym{)}  \ottsym{\{}  \ottnt{e}  \ottsym{\}}$, compilation is defined as follows:
\[
 \mathscr{C} ( \ottnt{P} ,  \mathit{c} , \ottmv{m} )  =\;   \mathscr{A} ( \ottnt{e} )   ;  \ottkw{Ret} 
\]

The $ \mathscr{A} $ function is recursively defined as presented
in figure \ref{fig:s2i}; we allow ourselves to refer to $\ottnt{P}$,
$\mathit{c}$ and $\ottmv{m}$ in this definition.
We denote by $\ottnt{P}  \ottsym{;}  \mathit{c}  \ottsym{,}  \ottmv{m}  \vdash  \ottnt{e}  \ottsym{:}  \mathit{c'}$ the predicate indicating that
expression $\ottnt{e}$ has type $\mathit{c'}$ when typed within method $\ottmv{m}$
of class $\mathit{c}$ from program $P$.
In the compilation of method calls, we assume a type inference
algorithm which, given $\ottnt{P}$, $\mathit{c}$ and $\ottmv{m}$, finds the
unique type $\mathit{c'}$ such that $\ottnt{P}  \ottsym{;}  \mathit{c}  \ottsym{,}  \ottmv{m}  \vdash  \ottnt{e}  \ottsym{:}  \mathit{c'}$.
We use it to statically resolve method calls by annotating
intermediate-level call instructions.
In this document, we do not present the type system nor the
type inference algorithm for our source language, which are
standard.

The invariant used by the compilation is that executing
$ \mathscr{A} ( \ottnt{e} ) $ in the intermediate level will either
diverge---when evaluating $\ottnt{e}$ in the high-level diverges--- or
terminate with exactly one extra object on the local stack---in which
case this object is exactly the result of evaluating $\ottnt{e}$.
In a method body, this object on top of the stack can then be returned
as the result of the method call, which is why $ \mathscr{A} ( \ottnt{e} ) $ is
followed by a $\ottkw{Ret}$ instruction in the main compilation
function $ \mathscr{C} $.

With this invariant in mind, the translation is rather
straightforward, which is not surprising since our abstract stack
machine was designed with this goal in mind.
An important point is that we keep the evaluation order of the source
language: left to right.
It matters because of side effects and early termination being
available in our language.

Let us explain the only non-trivial expression to compile: the
object identity test $\ottsym{(}  \ottnt{e_{{\mathrm{1}}}}  \ottsym{==}  \ottnt{e_{{\mathrm{2}}}}  \;?\;  \ottnt{e_{{\mathrm{3}}}}  \ottsym{:}  \ottnt{e_{{\mathrm{4}}}}  \ottsym{)}$, for which we use two
branching instructions $\ottkw{Skeq} \, \ottsym{(}  \ottsym{\mbox{$\mid$}}   \mathscr{A} ( \ottnt{e_{{\mathrm{4}}}} )   \ottsym{\mbox{$\mid$}}  \ottsym{+}  \ottsym{1}  \ottsym{)}$ and
$\ottkw{Skip} \, \ottsym{\mbox{$\mid$}}   \mathscr{A} ( \ottnt{e_{{\mathrm{3}}}} )   \ottsym{\mbox{$\mid$}}$.
Here, we denote by $\ottsym{\mbox{$\mid$}}   \mathscr{A} ( \ottnt{e} )   \ottsym{\mbox{$\mid$}}$ the length
of the sequence of instructions $ \mathscr{A} ( \ottnt{e} ) $.
With equal objects, executing $\ottkw{Skeq} \, \ottsym{(}  \ottsym{\mbox{$\mid$}}   \mathscr{A} ( \ottnt{e_{{\mathrm{4}}}} )   \ottsym{\mbox{$\mid$}}  \ottsym{+}  \ottsym{1}  \ottsym{)}$
will skip the code corresponding to $\ottnt{e_{{\mathrm{4}}}}$ and the
$\ottkw{Skip} \, \ottsym{\mbox{$\mid$}}   \mathscr{A} ( \ottnt{e_{{\mathrm{3}}}} )   \ottsym{\mbox{$\mid$}}$ instruction, hence branching
directly to the code corresponding to $\ottnt{e_{{\mathrm{3}}}}$ to execute it.
With different objects, executing
$\ottkw{Skeq} \, \ottsym{(}  \ottsym{\mbox{$\mid$}}   \mathscr{A} ( \ottnt{e_{{\mathrm{4}}}} )   \ottsym{\mbox{$\mid$}}  \ottsym{+}  \ottsym{1}  \ottsym{)}$ will do nothing and execution
will proceed with the code corresponding to
$\ottnt{e_{{\mathrm{4}}}}$, followed by a $\ottkw{Skip} \, \ottsym{\mbox{$\mid$}}   \mathscr{A} ( \ottnt{e_{{\mathrm{3}}}} )   \ottsym{\mbox{$\mid$}}$ instruction
which will unconditionally skip the code corresponding to
$\ottnt{e_{{\mathrm{3}}}}$, hence branching to the $\ottkw{Nop}$ instruction.
The effect of all this is that in the end,
we have executed $\ottnt{e_{{\mathrm{1}}}}$ and $\ottnt{e_{{\mathrm{2}}}}$ in this
order, popped the resulting objects from the stack, and either
executed $\ottnt{e_{{\mathrm{3}}}}$ if they were equal or $\ottnt{e_{{\mathrm{4}}}}$ if they were
different:
We execute the appropriate code in both cases.

\subsubsection{From Intermediate to Target}
\label{sec:i2t}


We now present our unoptimizing, type-preserving translation from
intermediate-machine compartments to target-level components.
Target-level compartments are defined as sets of untagged symbolic
machine memory regions.
Like in the source to intermediate compilation, the translation itself
is rather standard.
The exception is that components cannot trust other components to
follow conventions such as not messing with a global call stack, or
not modifying some registers.
As a consequence, components use local stacks and all relevant
registers need to be (re)initialized when the compiled component takes
control.
Other than that, all security is enforced by means of
instruction-level monitoring (\autoref{sec:micro-policy}).

\paragraph{Object Compilation}

Each object $\mathit{o}$ that was assigned a definition
$\ottsym{(}  \mathit{o_{{\mathrm{1}}}}  \ottsym{,} \, ... \, \ottsym{,}  \mathit{o_{\ottmv{l}}}  \ottsym{)}$ now gets its own region in target memory.
This region is assigned symbolic location $\ottkw{objl} \, \mathit{o}$ and spans over
$\ottmv{l}$ memory cells.
These cells are filled with the $\ottkw{objl} \, \mathit{o_{{\mathrm{1}}}}, \dots, \ottkw{objl} \, \mathit{o_{\ottmv{l}}}$
symbolic locations --- which are the addresses of these objects in
memory.
%

\paragraph{Local Stack Compilation}


Each local stack also gets its own memory region, under symbolic
location $\ottkw{stackl} \, \mathit{c}$ where $\mathit{c}$ is the name of the compartment
being compiled.
Components will maintain the following invariant during computation:
The first cell holds a symbolic pointer to the top of the stack, which
is $\ottsym{(}  \ottkw{stackl} \, \mathit{c}  \ottsym{)}  \ottsym{+}  \ottmv{l}$ where $\ottmv{l}$ is the length of the stack.
The following cells are used for storing actual values on the stack.

Here, we only care about compiling intermediate-level components that
come from the source to intermediate compiler.
For these components, the initial stack is always empty.
Hence, we just initialize the first cell to $\ottkw{stackl} \, \mathit{c}$, and the
initial content of extra cells can be arbitrary constants: their only
purpose is to be available for filling during the computation.

\paragraph{Method Compilation: From Coarse- to Fine-Grained Instructions}

\begin{figure}
\centering
\begin{align*}
&  \mathscr{A} ( \mathit{c} , \ottkw{Call} \, \mathit{c'} \, \ottmv{m}  \ottsym{;}  \ottnt{Icode} )  =\; \\
&~~ \text{\texttt{(* pop call argument and object *)}} \\
&~~\ottkw{Load} \, {r_{\textsf{sp} } } \, { {r_{\textsf{aux} } }_2 }  \ottsym{;}  \ottkw{Sub} \, {r_{\textsf{sp} } } \, {r_{\textsf{one} } } \, {r_{\textsf{sp} } }  \ottsym{;}  \ottkw{Load} \, {r_{\textsf{sp} } } \, { {r_{\textsf{aux} } }_1 }; \\
&~~ \text{\texttt{(* push current object and argument *)}} \\
&~~\ottkw{Store} \, {r_{\textsf{sp} } } \, {r_{\textsf{tgt} } }  \ottsym{;}  \ottkw{Add} \, {r_{\textsf{sp} } } \, {r_{\textsf{one} } } \, {r_{\textsf{sp} } }  \ottsym{;}  \ottkw{Store} \, {r_{\textsf{sp} } } \, {r_{\textsf{arg} } }; \\
&~~ \text{\texttt{(* save stack pointer *)}} \\
&~~\ottkw{Store} \, {r_{\textsf{spp} } } \, {r_{\textsf{sp} } }; \\
&~~ \text{\texttt{(* set call object and argument *)}} \\
&~~\ottkw{Mov} \, { {r_{\textsf{aux} } }_1 } \, {r_{\textsf{tgt} } }  \ottsym{;}  \ottkw{Mov} \, { {r_{\textsf{aux} } }_2 } \, {r_{\textsf{arg} } }; \\
&~~ \text{\texttt{(* perform call *)}} \\
&~~\ottkw{Const} \, \ottsym{(}  \ottkw{methl} \, \mathit{c'} \, \ottmv{m}  \ottsym{)} \, { {r_{\textsf{aux} } }_3 }  \ottsym{;}  \ottkw{Jal} \, { {r_{\textsf{aux} } }_3 }; \\
&~~ \text{\texttt{(* reinitialize environment *)}} \\
&~~\ottkw{Const} \, \ottsym{1} \, {r_{\textsf{one} } }  \ottsym{;}  \ottkw{Const} \, \ottsym{(}  \ottkw{stackl} \, \mathit{c}  \ottsym{)} \, {r_{\textsf{spp} } }  \ottsym{;}  \ottkw{Load} \, {r_{\textsf{spp} } } \, {r_{\textsf{sp} } }; \\
&~~ \text{\texttt{(* restore current object and argument *)}} \\
&~~\ottkw{Load} \, {r_{\textsf{sp} } } \, {r_{\textsf{arg} } }  \ottsym{;}  \ottkw{Sub} \, {r_{\textsf{sp} } } \, {r_{\textsf{one} } } \, {r_{\textsf{sp} } }  \ottsym{;}  \ottkw{Load} \, {r_{\textsf{sp} } } \, {r_{\textsf{tgt} } } \\
&~~ \text{\texttt{(* push call result *)}} \\
&~~\ottkw{Store} \, {r_{\textsf{sp} } } \, { r_{\textsf{ret} } }  \ottsym{;}   \mathscr{A} ( \mathit{c} , \ottnt{Icode} )  \\
\\
& \mathscr{A} ( \mathit{c} , \ottkw{Ret}  \ottsym{;}  \ottnt{Icode} )  =\; \\
&~~ \text{\texttt{(* pop return value *)}} \\
&~~\ottkw{Load} \, {r_{\textsf{sp} } } \, { r_{\textsf{ret} } }  \ottsym{;}  \ottkw{Sub} \, {r_{\textsf{sp} } } \, {r_{\textsf{one} } } \, {r_{\textsf{sp} } } ; \\
&~~ \text{\texttt{(* pop return address *)}} \\
&~~\ottkw{Load} \, {r_{\textsf{sp} } } \, {r_{\textsf{a} } }  \ottsym{;}  \ottkw{Sub} \, {r_{\textsf{sp} } } \, {r_{\textsf{one} } } \, {r_{\textsf{sp} } } ; \\
&~~ \text{\texttt{(* save stack pointer *)}} \\
&~~\ottkw{Store} \, {r_{\textsf{spp} } } \, {r_{\textsf{sp} } } ; \\
&~~ \text{\texttt{(* perform return *)}} \\
&~~\ottkw{Jump} \, {r_{\textsf{a} } }  \ottsym{;}   \mathscr{A} ( \mathit{c} , \ottnt{Icode} )  \\
\end{align*}
\caption{Compilation of communication-related instructions of the
  intermediate machine}
\label{fig:i2tcomm}
\end{figure}

\begin{figure}
\centering
\begin{align*}
&  \mathscr{A} ( \mathit{c} , \ottkw{This}  \ottsym{;}  \ottnt{Icode} )  =\; \text{\texttt{(* push object *)}} \\
&~~ \ottkw{Add} \, {r_{\textsf{sp} } } \, {r_{\textsf{one} } } \, {r_{\textsf{sp} } }  \ottsym{;}  \ottkw{Store} \, {r_{\textsf{sp} } } \, {r_{\textsf{tgt} } }  \ottsym{;}   \mathscr{A} ( \mathit{c} , \ottnt{Icode} )  \\
\\
&  \mathscr{A} ( \mathit{c} , \ottkw{Arg}  \ottsym{;}  \ottnt{Icode} )  =\; \text{\texttt{(* push argument *)}} \\
&~~ \ottkw{Add} \, {r_{\textsf{sp} } } \, {r_{\textsf{one} } } \, {r_{\textsf{sp} } }  \ottsym{;}  \ottkw{Store} \, {r_{\textsf{sp} } } \, {r_{\textsf{arg} } }  \ottsym{;}   \mathscr{A} ( \mathit{c} , \ottnt{Icode} )  \\
\\
&  \mathscr{A} ( \mathit{c} , \ottkw{Ref} \, \mathit{o}  \ottsym{;}  \ottnt{Icode} )  =\; \text{\texttt{(* push object o *)}} \\
&~~ \ottkw{Const} \, \ottsym{(}  \ottkw{objl} \, \mathit{o}  \ottsym{)} \, { {r_{\textsf{aux} } }_1 } ; \\
&~~ \ottkw{Add} \, {r_{\textsf{sp} } } \, {r_{\textsf{one} } } \, {r_{\textsf{sp} } }  \ottsym{;}  \ottkw{Store} \, {r_{\textsf{sp} } } \, { {r_{\textsf{aux} } }_1 }  \ottsym{;}   \mathscr{A} ( \mathit{c} , \ottnt{Icode} )  \\
\\
&  \mathscr{A} ( \mathit{c} , \ottkw{Drop}  \ottsym{;}  \ottnt{Icode} )  =\; \ottkw{Sub} \, {r_{\textsf{sp} } } \, {r_{\textsf{one} } } \, {r_{\textsf{sp} } }  \ottsym{;}   \mathscr{A} ( \mathit{c} , \ottnt{Icode} )  \\
\\
&  \mathscr{A} ( \mathit{c} , \ottkw{Sel} \, \ottmv{f}  \ottsym{;}  \ottnt{Icode} )  =\; \\
&~~ \ottkw{Const} \, \ottsym{(}  \ottmv{f}  \ottsym{-}  \ottsym{1}  \ottsym{)} \, { {r_{\textsf{aux} } }_2 }; \\
&~~ \text{\texttt{(* pop object to select from *)}} \\
&~~ \ottkw{Load} \, {r_{\textsf{sp} } } \, { {r_{\textsf{aux} } }_1 }  \ottsym{;}  \ottkw{Add} \, { {r_{\textsf{aux} } }_1 } \, { {r_{\textsf{aux} } }_2 } \, { {r_{\textsf{aux} } }_1 } ; \\
&~~ \text{\texttt{(* load and push field value *)}} \\
&~~ \ottkw{Load} \, { {r_{\textsf{aux} } }_1 } \, { {r_{\textsf{aux} } }_1 }  \ottsym{;}  \ottkw{Store} \, {r_{\textsf{sp} } } \, { {r_{\textsf{aux} } }_1 }  \ottsym{;}   \mathscr{A} ( \mathit{c} , \ottnt{Icode} )  \\
\\
&  \mathscr{A} ( \mathit{c} , \ottkw{Upd} \, \ottmv{f}  \ottsym{;}  \ottnt{Icode} )  =\; \\
&~~ \ottkw{Const} \, \ottsym{(}  \ottmv{f}  \ottsym{-}  \ottsym{1}  \ottsym{)} \, { {r_{\textsf{aux} } }_2 } ; \\
&~~ \text{\texttt{(* pop new field value and object *)}} \\
&~~ \ottkw{Load} \, {r_{\textsf{sp} } } \, { {r_{\textsf{aux} } }_3 }  \ottsym{;}  \ottkw{Sub} \, {r_{\textsf{sp} } } \, {r_{\textsf{one} } } \, {r_{\textsf{sp} } }  \ottsym{;}  \ottkw{Load} \, {r_{\textsf{sp} } } \, { {r_{\textsf{aux} } }_1 } ; \\
&~~ \text{\texttt{(* perform update on object *)}} \\
&~~ \ottkw{Add} \, { {r_{\textsf{aux} } }_1 } \, { {r_{\textsf{aux} } }_2 } \, { {r_{\textsf{aux} } }_1 }  \ottsym{;}  \ottkw{Store} \, { {r_{\textsf{aux} } }_1 } \, { {r_{\textsf{aux} } }_3 } ; \\
&~~ \text{\texttt{(* push new field value *)}} \\
&~~ \ottkw{Store} \, {r_{\textsf{sp} } } \, { {r_{\textsf{aux} } }_3 }  \ottsym{;}   \mathscr{A} ( \mathit{c} , \ottnt{Icode} )  \\
\end{align*}
\caption{Compilation of stack-related instructions of the intermediate machine}
\label{fig:i2tstack}
\end{figure}

\begin{figure}
\centering
\begin{align*}
&  \mathscr{A} ( \mathit{c} , \ottkw{Nop}  \ottsym{;}  \ottnt{Icode} )  =\;
  \ottkw{Nop}  \ottsym{;}   \mathscr{A} ( \mathit{c} , \ottnt{Icode} )  \\
\\
&  \mathscr{A} ( \mathit{c} , \ottkw{Halt}  \ottsym{;}  \ottnt{Icode} )  =\; \ottkw{Halt}  \ottsym{;}   \mathscr{A} ( \mathit{c} , \ottnt{Icode} )  \\
\\
&  \mathscr{A} ( \mathit{c} ,   \ottkw{Skip} \, \ottmv{k}  ;  \ottnt{Iinstr_{{\mathrm{1}}}}  \ottsym{;} \, ... \, \ottsym{;}  \ottnt{Iinstr_{\ottmv{k}}}   ;  \ottnt{Icode}  )  =\; \\
&~~\ottkw{Bnz} \, {r_{\textsf{one} } } \,  \mathscr{L} ( \ottnt{Iinstr_{{\mathrm{1}}}}  \ottsym{;} \, ... \, \ottsym{;}  \ottnt{Iinstr_{\ottmv{k}}} )  ; \\
&~~ \mathscr{A} ( \mathit{c} ,  \ottnt{Iinstr_{{\mathrm{1}}}}  \ottsym{;} \, ... \, \ottsym{;}  \ottnt{Iinstr_{\ottmv{k}}}  ;  \ottnt{Icode}  )  \\
\\
&  \mathscr{A} ( \mathit{c} ,   \ottkw{Skeq} \, \ottmv{k}  ;  \ottnt{Iinstr_{{\mathrm{1}}}}  \ottsym{;} \, ... \, \ottsym{;}  \ottnt{Iinstr_{\ottmv{k}}}   ;  \ottnt{Icode}  )  =\; \\
&~~ \text{\texttt{(* pop and compare objects *)}} \\
&~~\ottkw{Load} \, {r_{\textsf{sp} } } \, { {r_{\textsf{aux} } }_2 }  \ottsym{;}  \ottkw{Sub} \, {r_{\textsf{sp} } } \, {r_{\textsf{one} } } \, {r_{\textsf{sp} } }  \ottsym{;}  \ottkw{Load} \, {r_{\textsf{sp} } } \, { {r_{\textsf{aux} } }_1 } ; \\
&~~\ottkw{Sub} \, {r_{\textsf{sp} } } \, {r_{\textsf{one} } } \, {r_{\textsf{sp} } }  \ottsym{;}   { \ottkw{Binop} _{ \ottsym{=} } }  { {r_{\textsf{aux} } }_1 }   { {r_{\textsf{aux} } }_2 }   { {r_{\textsf{aux} } }_1 }  ; \\
&~~ \text{\texttt{(* branch according to result *)}} \\
&~~\ottkw{Bnz} \, { {r_{\textsf{aux} } }_1 } \,  \mathscr{L} ( \ottnt{Iinstr_{{\mathrm{1}}}}  \ottsym{;} \, ... \, \ottsym{;}  \ottnt{Iinstr_{\ottmv{k}}} )  ; \\
&~~ \mathscr{A} ( \mathit{c} ,  \ottnt{Iinstr_{{\mathrm{1}}}}  \ottsym{;} \, ... \, \ottsym{;}  \ottnt{Iinstr_{\ottmv{k}}}  ;  \ottnt{Icode}  ) 
\end{align*}
where $ \mathscr{L} ( \ottnt{Iinstr_{{\mathrm{1}}}}  \ottsym{;} \, ... \, \ottsym{;}  \ottnt{Iinstr_{\ottmv{k}}} ) $ is the length of
the sequence of compiled instructions corresponding to instructions
$\ottnt{Iinstr_{{\mathrm{1}}}}  \ottsym{;} \, ... \, \ottsym{;}  \ottnt{Iinstr_{\ottmv{k}}}$, defined as:
\begin{align*}
 \mathscr{L} ( \text{\texttt{[]} } )  =&\; 0 \\
 \mathscr{L} ( \ottkw{Drop}  \ottsym{;}  \ottnt{Icode} )  =\;
 \mathscr{L} ( \ottkw{Nop}  \ottsym{;}  \ottnt{Icode} )  =&\;
1 +  \mathscr{L} ( \ottnt{Icode} )  \\
 \mathscr{L} ( \ottkw{Halt}  \ottsym{;}  \ottnt{Icode} )  =\;
 \mathscr{L} ( \ottkw{Skip} \, \mathit{n}  \ottsym{;}  \ottnt{Icode} )  =&\;
1 +  \mathscr{L} ( \ottnt{Icode} )  \\
 \mathscr{L} ( \ottkw{This}  \ottsym{;}  \ottnt{Icode} )  =\;
 \mathscr{L} ( \ottkw{Arg}  \ottsym{;}  \ottnt{Icode} )  =&\;
2 +  \mathscr{L} ( \ottnt{Icode} )  \\
 \mathscr{L} ( \ottkw{Ref} \, \mathit{o}  \ottsym{;}  \ottnt{Icode} )  =&\;
3 +  \mathscr{L} ( \ottnt{Icode} )  \\
 \mathscr{L} ( \ottkw{Sel} \, \ottmv{f}  \ottsym{;}  \ottnt{Icode} )  =&\;
  5 +  \mathscr{L} ( \ottnt{Icode} )  \\
 \mathscr{L} ( \ottkw{Ret}  \ottsym{;}  \ottnt{Icode} )  =\;
 \mathscr{L} ( \ottkw{Skeq} \, \mathit{n}  \ottsym{;}  \ottnt{Icode} )  =&\;
  6 +  \mathscr{L} ( \ottnt{Icode} )  \\
 \mathscr{L} ( \ottkw{Upd} \, \ottmv{f}  \ottsym{;}  \ottnt{Icode} )  =&\;
  7 +  \mathscr{L} ( \ottnt{Icode} )  \\
 \mathscr{L} ( \ottkw{Call} \, \mathit{c'} \, \ottmv{m}  \ottsym{;}  \ottnt{Icode} )  =&\;
  18 +  \mathscr{L} ( \ottnt{Icode} )  \\
\end{align*}
\caption{Compilation of
control-related instructions of the intermediate machine}
\label{fig:i2tcontrol}
\end{figure}

A method with index $m$ gets its own memory region under symbolic
location $\ottkw{methl} \, \mathit{c} \, \ottmv{m}$ where $\mathit{c}$ is the name of the compartment
being compiled.
The length of these memory regions is that of the corresponding
compiled code, which is what they are filled with.
The compilation is a translation of the method bodies, mapping
each intermediate-level instruction to a low-level
instruction sequence.

%

The compilation uses ten distinct general-purpose registers.
Register ${r_{\textsf{a} } }$ is automatically filled upon low-level call
instructions $\ottkw{Jal}$ --- following the semantics of the machine
studied in the original micro-policies paper~\cite{micropolicies2015} --- for the
callees to get the address to which they should return.
Registers ${r_{\textsf{tgt} } }$, ${r_{\textsf{arg} } }$ and ${ r_{\textsf{ret} } }$ are used for value
communication: ${r_{\textsf{tgt} } }$ stores the object on which we're calling a
method---the \emph{target object}---and ${r_{\textsf{arg} } }$ the argument for
the method, while ${ r_{\textsf{ret} } }$ stores the result of a call on a return.
%
%
Registers ${ {r_{\textsf{aux} } }_1 }$, ${ {r_{\textsf{aux} } }_2 }$, ${ {r_{\textsf{aux} } }_3 }$ are used for storing
temporary results.
Register ${r_{\textsf{sp} } }$ holds a pointer to the current top value of the
local stack --- we call this register the \emph{stack pointer} register.
Register ${r_{\textsf{spp} } }$ always holds a pointer to a fixed location where
the stack pointer can be stored and restored -- this location is the
first cell in the memory region dedicated to the local stack.
Finally, register ${r_{\textsf{one} } }$ always stores the value $1$ so that this
particular value is always easily available.

The compilation of method $\ottmv{m}$ of class $\mathit{c}$ with method body
$\ottnt{Icode}$ is defined as follows:
\begin{align*}
& \mathscr{C} ( \mathit{c} ,  \ottmv{m} ,  \ottnt{Icode} )  = \\
&~~ \ottkw{Const} \, \ottsym{1} \, {r_{\textsf{one} } } ; \\
&~~ \text{\texttt{(* load stack pointer *)}} \\
&~~\ottkw{Const} \, \ottsym{(}  \ottkw{stackl} \, \mathit{c}  \ottsym{)} \, {r_{\textsf{spp} } }  \ottsym{;}  \ottkw{Load} \, {r_{\textsf{spp} } } \, {r_{\textsf{sp} } } ; \\
&~~ \text{\texttt{(* push return address *)}} \\
&~~\ottkw{Add} \, {r_{\textsf{sp} } } \, {r_{\textsf{one} } } \, {r_{\textsf{sp} } }  \ottsym{;}  \ottkw{Store} \, {r_{\textsf{sp} } } \, {r_{\textsf{a} } }  \ottsym{;}   \mathscr{A} ( \mathit{c} , \ottnt{Icode} ) 
\end{align*}
where $ \mathscr{A} ( \mathit{c} , \ottnt{Icode} ) $ is an auxiliary, recursively defined
function having $ \mathscr{A} ( \mathit{c} , \text{\texttt{[]} } )  =\; \text{\texttt{[]} }$ as a base
case.
As shown in the code snippet, the first instructions initialize the
registers for them to match with the invariant we just explained
informally.

Compilation is most interesting for calls and return
instructions, which we present in figure~\ref{fig:i2tcomm}.
We also present the compilation of stack-manipulating instructions in
figure~\ref{fig:i2tstack} and that of control-related instructions in
figure~\ref{fig:i2tcontrol}.
In all these figures, inline comments are provided so that the
interested reader can get a quick understanding of what is
happening.

More standard compilers typically use a global call stack and compile
code under the assumption that other components will not break
invariants (such as ${r_{\textsf{one} } }$ always holding value $1$) nor mess with
the call stack.
In our case, however, other components may be controlled by an
attacker, which is incompatible with such assumptions.
As a consequence, we use local stacks not only for intermediate
results, but also to spill registers ${r_{\textsf{tgt} } }$ and ${r_{\textsf{arg} } }$
before performing a call.
After the call, we restore the values of registers ${r_{\textsf{one} } }$,
${r_{\textsf{spp} } }$ and ${r_{\textsf{sp} } }$ that could have been overwritten by the
callee, then fill registers ${r_{\textsf{tgt} } }$ and ${r_{\textsf{arg} } }$ from the local
stack.

There is a lot of room for improvement in terms of compiler
efficiency; having a code optimization pass would be very interesting
in the future.

\subsection{Micro-Policy Protecting Abstractions}
\label{sec:micro-policy}


Here, we first present the abstractions that our source language
provides with respect to low-level machine code, and give intuition
about how we can protect them using a micro-policy.
Then, we present our actual micro-policy and explain how it
effectively implements this intuition.
%
We end with a description of the low-level loader, which sets the
initial memory tags for the program.

\subsubsection{Enforcing Class Isolation via Compartmentalization}
\label{sec:classisolation}

\paragraph{Abstraction}

Because in the source language fields are private, classes have no way
to read or to write to the data of other classes.
Moreover, classes cannot read or write code, which is fixed.
Finally, the only way to communicate between classes is to perform a
method call.

In machine code, however, $\ottkw{Load}$, $\ottkw{Store}$ and $\ottkw{Jump}$
operations can target any address in memory, including those of other
components.
This interaction must be restricted to preserve the class isolation
abstraction.

\paragraph{Protection Mechanism}

To enforce class isolation, memory cells and the program counter get
tagged with a class name, which can be seen as a color.
The code and data belonging to class $\mathit{c}$ get tagged with color $\mathit{c}$.

$\ottkw{Load}$ and $\ottkw{Store}$ instructions are restricted to locations
having the same color as the current instruction.
Moreover, the rules will compare the next instruction's compartment
to that of the current instruction:
Switching compartments is only allowed for $\ottkw{Jump}$ and $\ottkw{Jal}$;
however, we need \emph{more} protection on these instructions because
switching compartments should be further restricted: this is what we now
present.

\subsubsection{Enforcing Method Call Discipline using Method Entry
  Points, Linear Return Capabilities, and Register Cleaning}
\label{sec:calldiscipline}



\paragraph{Abstraction}

In the source language,
callers and callees obey a strict call discipline: a caller performs a
method call, which leads to the execution of the method body, and once
this execution ends evaluation proceeds with the next operation of the caller.
In machine code, though, the $\ottkw{Jal}$ and $\ottkw{Jump}$ instructions can
target any address.

Moreover, in the high-level language callers and callees give no
information to each other except for the arguments and the result of
the call.
In the low-level machine, however, registers may carry extra
information about the caller or the callee and their intermediate
computational results.
This suggests a need for \emph{register cleaning} upon calls and
returns.

\paragraph{Protection Mechanism}

On our machine, calls are done via $\ottkw{Jal}$ instructions, which store
the return address in the ${r_{\textsf{a} } }$ register, and returns by executing
a $\ottkw{Jump}$ to the value stored in ${r_{\textsf{a} } }$.
The first goal here is to ensure that a caller can only performs calls
at method entry points, and that a callee can only return to the
return address it was given in register ${r_{\textsf{a} } }$ on the corresponding call.

To this end, we extend the memory tags to allow tagging a memory
location as a method entry point.
Then, we monitor $\ottkw{Jal}$ instructions so that they can only target
such locations.
This is enough to protect method calls.
For returns, however, the problem is not that simple.
In contrast with calls, for which all method entry points are equally
valid, only one return address is the good one at any given point in
time; it is the one that was transferred to the callee in the ${r_{\textsf{a} } }$
register by the corresponding call.
More precisely, because there can be nested calls, there is exactly
one valid return address for each call depth.

We reflect this by tracking the current call depth in the PC tag: it starts at
zero and we increment it on $\ottkw{Jal}$ instructions and decrement it on
$\ottkw{Jump}$ instructions \chrev{that correspond to returns}.
With such tracking, we can now, upon a $\ottkw{Jal}$ instruction, tag
register ${r_{\textsf{a} } }$ to mark its content as the only valid return address
for the current call depth.
It is however crucial to make sure that when we go back to call depth
$\mathit{n}$, there isn't any return capability for call depth $\mathit{n}+1$
idling in the system.
We do this by enforcing that the tag on the return address is never
duplicated, which makes it a \emph{linear return capability}.
The return capability gets created upon $\ottkw{Jal}$, moved from register
to memory upon $\ottkw{Store}$ and from memory to register upon
$\ottkw{Load}$, moved between registers upon $\ottkw{Mov}$, and is finally
destroyed upon returning via $\ottkw{Jump}$.
When it gets destroyed, we can infer that the capability has
disappeared from the system, since there was always only one.


Now that our first goal is met, we can think about the second one:
ensuring that compiled components do not transmit more information
upon call and return than they do within the high-level semantics.
Upon return, the distrusted caller potentially receives more
information than in the high-level:
Uncleaned registers could hold information from the compiled callee.
Upon calls, the distrusted callee similarly receives information
through registers, but has also an extra information to use:
the content of register ${r_{\textsf{a} } }$, which holds the return address.
This content leaks the identity of the compiled caller to the
distrusted callee, while there is no way for a callee to know the
identity of its caller in the high-level.
%
Fortunately, the content of register ${r_{\textsf{a} } }$ is already a
specifically tagged value, and we already prevent the use of linear
return capabilities for any other means than returning through it or
moving it around.

Let us now review the general purpose registers which could leak
information about our compiled partial programs.
${r_{\textsf{tgt} } }$ and ${r_{\textsf{arg} } }$ could leak the current object and
argument to the distrusted caller upon return, but this is fine: the
caller was the one who set them before the call, so he already knows
them.
Upon call, these registers do not leak information either since,
according to the compilation scheme, they are already overwritten with
call parameters.
${ r_{\textsf{ret} } }$ could leak a previous result value of the compiled caller
to the distrusted callee upon call: it has to be cleaned.
Upon return, however, and according to the compilation scheme, this
register is already overwritten with the return value for the call.
${ {r_{\textsf{aux} } }_1 }$, ${ {r_{\textsf{aux} } }_2 }$, ${ {r_{\textsf{aux} } }_3 }$ could leak intermediate results
of the computation upon return and have to be cleaned accordingly.
Upon call, however, following the compilation scheme, they are already
overwritten with information that the distrusted callee either will
get or already has.
${r_{\textsf{spp} } }$ could leak the identity of the compiled caller to the
distrusted callee upon call, since it contains a pointer to the
caller's local stack's memory region: it has to be cleaned.
In the case of a return however, the identity of the compiled callee
is already known by the distrusted caller so no new information will
be leaked.
%
${r_{\textsf{sp} } }$ could leak information about the current state of the local
stack, as well as the identity of the compiled caller, and should
accordingly be cleaned in both call and return cases.
${r_{\textsf{one} } }$ will be known by the distrusted caller or the distrusted
callee to always hold value $1$, and thus won't leak any
information.
%
%
${r_{\textsf{a} } }$ is already protected from being leaked: it is overwritten
both at call and return time.
In the case of a call, it is overwritten upon the execution of the
$\ottkw{Jal}$ instruction to hold the new return address.
In the case of a return, according to the compilation scheme, it will
be overwritten before performing the return to hold the address to which we are
returning.
%
%
This description concerns the current unoptimizing compiler; for an
optimizing compiler the situation would be a little different: more
information could be leaked, and accordingly more registers would have
to be cleaned.

A first solution would be to have the compiler produce register reset
instructions $\ottkw{Const} \, \ottsym{0} \, \mathit{r}$ for every register $\mathit{r}$ that could
leak information, before any external call or return instruction.
However, this would be very expensive.
This is one of the reasons why we have made the following assumption
about our target symbolic micro-policy machine:
The tags of some fixed registers (here, ${ r_{\textsf{ret} } }$, ${r_{\textsf{spp} } }$ and
${r_{\textsf{sp} } }$ upon $\ottkw{Jal}$, and ${ {r_{\textsf{aux} } }_1 }$, ${ {r_{\textsf{aux} } }_2 }$, ${ {r_{\textsf{aux} } }_3 }$
and ${r_{\textsf{sp} } }$ upon $\ottkw{Jump}$) can be updated in our symbolic
micro-policy rules.
We are thus by assumption able to clean the registers that might leak
information, by using a special tag to mark these registers as cleared
when we execute a $\ottkw{Jump}$ or a $\ottkw{Jal}$ instruction.

\subsubsection{Enforcing Type Safety Dynamically}
\label{sec:typesafety}

\paragraph{Abstraction}

Finally, in the source language callees and callers expect arguments
or return values that are well-typed with respect to method
signatures:
We only consider high-level components that are statically
well-typed and thus have to comply with the type interface they
declare.
At the machine code level, however, compromised components are untyped
and can feed arbitrary machine words to uncompromised ones, without
any a priori typing restrictions.

\paragraph{Protection Mechanism}

We use micro-policy rules to ensure that method arguments and
return values always have the type mentioned in type signatures.
Fortunately, our type system is simple enough that we can encode types
(which are just class names) as tags.
Hence, we can build upon the call discipline mechanism above and add the
expected argument type to the entry point tag, and the expected return
type to the linear return capability tag.

Our dynamic typing mechanism relies on the loader to initialize memory
tags appropriately.
This initial tagging will be presented in detail after the
micro-policy itself:
The main idea is that a register or memory location holding an object
pointer will be tagged as such, together with the type of the object.
This dynamic type information is moved around with the value when it
is stored, loaded or copied.
One key for making this possible is that the $\ottkw{Const} \, \ottsym{(}  \ottkw{objl} \, \mathit{o}  \ottsym{)} \, \mathit{r}$ instructions
which put an object reference in a register are \emph{blessed} with
the type of this object according to the type declared for $o$:
executing a blessed $\ottkw{Const}$ instruction will put the corresponding
type information on the target register's tag.
%

Remember that we assume that we can check the next instruction tag in
micro-policy rules.
With dynamic type information available, we can then do type checking
by looking at the tags of registers ${r_{\textsf{tgt} } }$ and ${r_{\textsf{arg} } }$ upon
call, and that of register ${ r_{\textsf{ret} } }$ upon return.
Upon call, we will compare with type information from the next
instruction tag, which is the tag for the method entry point.
Upon return, we will compare with type information from the return
capability's tag.
For these checks to be possible, we use the following assumption we
made about our target symbolic micro-policy machine:
The tags of some fixed registers (here, ${r_{\textsf{tgt} } }$ and ${r_{\textsf{arg} } }$
upon $\ottkw{Jal}$ and ${ r_{\textsf{ret} } }$ upon $\ottkw{Jump}$) can be checked in our
symbolic micro-policy rules.




\subsubsection{Micro-Policy in Detail}

As presented in section~\ref{sec:target}, a symbolic micro-policy is
the combination of a collection of symbolic tags and a transfer
function.
We first detail our tag syntax.
Then, we give the rules of our symbolic micro-policy, which define
its transfer function~\cite{micropolicies2015}.
Finally, we explain how the
loader initially tags program memory following the
program's export declarations.

\paragraph{Symbolic Micro-Policy Tags}

\begin{figure}
\ottgrammartabular{
\otttpc\ottinterrule
\otttmem\ottinterrule
\otttreg\ottinterrule
\ottbt\ottinterrule
\ottet\ottinterrule
\ottvt\ottinterrule
\ottrt\ottafterlastrule
}
\caption{Symbolic micro-policy tags syntax}
\label{fig:tags}
\end{figure}

Our collection of symbolic micro-policy tags is presented in
figure \ref{fig:tags}.

The tag on the program counter register is a natural number $\mathit{n}$
which represents the current call depth.

Memory tags $\ottsym{(}  \ottnt{bt}  \ottsym{,}  \mathit{c}  \ottsym{,}  \ottnt{et}  \ottsym{,}  \ottnt{vt}  \ottsym{)}$ combine the various information
we need about memory cells:
First, a memory cell belongs to a compartment $\mathit{c}$.
Its $\ottnt{bt}$ tag can be either $\textbf{B} \, \mathit{c'}$ for it to be blessed
with type $\mathit{c'}$ --- which means that it is a $\ottkw{Const}$
instruction which puts an object of type $\mathit{c'}$ in its target
register --- or $\textbf{NB}$ when it shouldn't be blessed.
Similarly, its $\ottnt{et}$ tag can be either $\textbf{EP} \, {c_{\textsf{a} } }  \rightarrow  {c_{\textsf{r} } }$
when it is the entry point of a method of type signature
${c_{\textsf{a} } }  \ottsym{(}  {c_{\textsf{r} } }  \ottsym{)}$, or $\textbf{NEP}$ when it is not.
Finally, the $\ottnt{vt}$ tag is for the content of the memory cell,
which can be either:
a cleared value, tagged $\bot$;
a return capability for going back to call depth $\mathit{n}$ by providing
a return value of type ${c_{\textsf{r} } }$, tagged $\ottkw{Ret} \, \mathit{n} \, {c_{\textsf{r} } }$;
an object pointer of type $\mathit{c'}$, tagged $\textbf{O} \, \mathit{c'}$;
or a regular word, tagged $\textbf{W}$.

Tags on registers are the same $\ottnt{vt}$ tags as the ones for the
content of memory cells:
The content of the register is also tagged as a cleared
value, a return capability, an object pointer, or a regular word.


\paragraph{Micro-Policy Rules}

\begin{figure*}
\centering
\[
\begin{array}{l}
\ottkw{Nop}  \ottsym{:}  \ottsym{\{}  {t_{\textsf{pc} } }  \ottsym{=}  \mathit{n}  \ottsym{,}  {t_{\textsf{ci} } }  \ottsym{=}  \ottsym{(}  \textbf{NB}  \ottsym{,}  \mathit{c}  \ottsym{,}  \ottnt{et}  \ottsym{,}  \textbf{W}  \ottsym{)}  \ottsym{,}  {t_{\textsf{ni} } }  \ottsym{=}  \ottsym{(}  \ottnt{bt'}  \ottsym{,}  \mathit{c}  \ottsym{,}  \ottnt{et'}  \ottsym{,}  \ottnt{rt'}  \ottsym{)}  \ottsym{\}}  \implies  \ottsym{\{}  {t_{\textsf{pc} } }'  \ottsym{=}  \mathit{n}  \ottsym{\}} \\
\ottkw{Const} \, i \, {r_{\textsf{d} } }  \ottsym{:}  \ottsym{\{}  {t_{\textsf{pc} } }  \ottsym{=}  \mathit{n}  \ottsym{,}  {t_{\textsf{ci} } }  \ottsym{=}  \ottsym{(}  \textbf{NB}  \ottsym{,}  \mathit{c}  \ottsym{,}  \ottnt{et}  \ottsym{,}  \textbf{W}  \ottsym{)}  \ottsym{,}  {t_{\textsf{ni} } }  \ottsym{=}  \ottsym{(}  \ottnt{bt'}  \ottsym{,}  \mathit{c}  \ottsym{,}  \ottnt{et'}  \ottsym{,}  \ottnt{rt'}  \ottsym{)}  \ottsym{\}}  \implies  \ottsym{\{}  {t_{\textsf{pc} } }'  \ottsym{=}  \mathit{n}  \ottsym{,}  {t_{r_{\textsf{d} } } }'  \ottsym{=}  \textbf{W}  \ottsym{\}} \\
\ottkw{Const} \, i \, {r_{\textsf{d} } }  \ottsym{:}  \ottsym{\{}  {t_{\textsf{pc} } }  \ottsym{=}  \mathit{n}  \ottsym{,}  {t_{\textsf{ci} } }  \ottsym{=}  \ottsym{(}  \textbf{B} \, \mathit{c'}  \ottsym{,}  \mathit{c}  \ottsym{,}  \ottnt{et}  \ottsym{,}  \textbf{W}  \ottsym{)}  \ottsym{,}  {t_{\textsf{ni} } }  \ottsym{=}  \ottsym{(}  \ottnt{bt'}  \ottsym{,}  \mathit{c}  \ottsym{,}  \ottnt{et'}  \ottsym{,}  \ottnt{rt'}  \ottsym{)}  \ottsym{\}}  \implies  \ottsym{\{}  {t_{\textsf{pc} } }'  \ottsym{=}  \mathit{n}  \ottsym{,}  {t_{r_{\textsf{d} } } }'  \ottsym{=}  \textbf{O} \, \mathit{c'}  \ottsym{\}} \\
\ottkw{Mov} \, {r_{\textsf{s} } } \, {r_{\textsf{d} } }  \ottsym{:}  \ottsym{\{}  {t_{\textsf{pc} } }  \ottsym{=}  \mathit{n}  \ottsym{,}  {t_{\textsf{ci} } }  \ottsym{=}  \ottsym{(}  \textbf{NB}  \ottsym{,}  \mathit{c}  \ottsym{,}  \ottnt{et}  \ottsym{,}  \textbf{W}  \ottsym{)}  \ottsym{,}  {t_{\textsf{ni} } }  \ottsym{=}  \ottsym{(}  \ottnt{bt'}  \ottsym{,}  \mathit{c}  \ottsym{,}  \ottnt{et'}  \ottsym{,}  \ottnt{rt'}  \ottsym{)}  \ottsym{,}  {t_{r_{\textsf{s} } } }  \ottsym{=}  \ottnt{vt}  \ottsym{,}  {t_{r_{\textsf{d} } } }  \ottsym{=}  \ottnt{vt'}  \ottsym{\}}  \\  \implies  \ottsym{\{}  {t_{\textsf{pc} } }'  \ottsym{=}  \mathit{n}  \ottsym{,}  {t_{r_{\textsf{d} } } }'  \ottsym{=}  \ottnt{vt}  \ottsym{,}  {t_{r_{\textsf{s} } } }'  \ottsym{=}  \ottkw{clear} \, \ottnt{vt}  \ottsym{\}} \\
 { \ottkw{Binop} _{ \ottnt{op} } }  \mathit{r_{{\mathrm{1}}}}   \mathit{r_{{\mathrm{2}}}}   {r_{\textsf{d} } }   \ottsym{:}  \ottsym{\{}  {t_{\textsf{pc} } }  \ottsym{=}  \mathit{n}  \ottsym{,}  {t_{\textsf{ci} } }  \ottsym{=}  \ottsym{(}  \textbf{NB}  \ottsym{,}  \mathit{c}  \ottsym{,}  \ottnt{et}  \ottsym{,}  \textbf{W}  \ottsym{)}  \ottsym{,}  {t_{\textsf{ni} } }  \ottsym{=}  \ottsym{(}  \ottnt{bt'}  \ottsym{,}  \mathit{c}  \ottsym{,}  \ottnt{et'}  \ottsym{,}  \ottnt{rt'}  \ottsym{)}  \ottsym{,}  {t_{\textsf{r} } }_{{\mathrm{1}}}  \ottsym{=}  \textbf{W}  \ottsym{,}  {t_{\textsf{r} } }_{{\mathrm{2}}}  \ottsym{=}  \textbf{W}  \ottsym{\}}  \\  \implies  \ottsym{\{}  {t_{\textsf{pc} } }'  \ottsym{=}  \mathit{n}  \ottsym{,}  {t_{r_{\textsf{d} } } }'  \ottsym{=}  \textbf{W}  \ottsym{\}}\\
\ottkw{Load} \, {r_{\textsf{p} } } \, {r_{\textsf{d} } }  \ottsym{:}  \ottsym{\{}  {t_{\textsf{pc} } }  \ottsym{=}  \mathit{n}  \ottsym{,}  {t_{\textsf{ci} } }  \ottsym{=}  \ottsym{(}  \textbf{NB}  \ottsym{,}  \mathit{c}  \ottsym{,}  \ottnt{et}  \ottsym{,}  \textbf{W}  \ottsym{)}  \ottsym{,}  {t_{\textsf{ni} } }  \ottsym{=}  \ottsym{(}  \ottnt{bt'}  \ottsym{,}  \mathit{c}  \ottsym{,}  \ottnt{et'}  \ottsym{,}  \ottnt{rt'}  \ottsym{)}  \ottsym{,}  {t_{r_{\textsf{p} } } }  \ottsym{=}  \textbf{W}  \ottsym{,}  {t_{\textsf{mem} } }  \ottsym{=}  \ottsym{(}  \ottnt{bt}  \ottsym{,}  \mathit{c}  \ottsym{,}  \ottnt{et''}  \ottsym{,}  \ottnt{vt}  \ottsym{)}  \ottsym{,}  {t_{r_{\textsf{d} } } }  \ottsym{=}  \ottnt{vt'}  \ottsym{\}}  \\  \implies  \ottsym{\{}  {t_{\textsf{pc} } }'  \ottsym{=}  \mathit{n}  \ottsym{,}  {t_{r_{\textsf{d} } } }'  \ottsym{=}  \ottnt{vt}  \ottsym{,}  {t_{\textsf{mem} } }'  \ottsym{=}  \ottsym{(}  \ottnt{bt}  \ottsym{,}  \mathit{c}  \ottsym{,}  \ottnt{et''}  \ottsym{,}  \ottkw{clear} \, \ottnt{vt}  \ottsym{)}  \ottsym{\}}\\
\ottkw{Store} \, {r_{\textsf{p} } } \, {r_{\textsf{s} } }  \ottsym{:}  \ottsym{\{}  {t_{\textsf{pc} } }  \ottsym{=}  \mathit{n}  \ottsym{,}  {t_{\textsf{ci} } }  \ottsym{=}  \ottsym{(}  \textbf{NB}  \ottsym{,}  \mathit{c}  \ottsym{,}  \ottnt{et}  \ottsym{,}  \textbf{W}  \ottsym{)}  \ottsym{,}  {t_{\textsf{ni} } }  \ottsym{=}  \ottsym{(}  \ottnt{bt'}  \ottsym{,}  \mathit{c}  \ottsym{,}  \ottnt{et'}  \ottsym{,}  \ottnt{rt'}  \ottsym{)}  \ottsym{,}  {t_{r_{\textsf{p} } } }  \ottsym{=}  \textbf{W}  \ottsym{,}  {t_{r_{\textsf{s} } } }  \ottsym{=}  \ottnt{vt}  \ottsym{,}  {t_{\textsf{mem} } }  \ottsym{=}  \ottsym{(}  \ottnt{bt}  \ottsym{,}  \mathit{c}  \ottsym{,}  \ottnt{et''}  \ottsym{,}  \ottnt{vt'}  \ottsym{)}  \ottsym{\}}  \\  \implies  \ottsym{\{}  {t_{\textsf{pc} } }'  \ottsym{=}  \mathit{n}  \ottsym{,}  {t_{\textsf{mem} } }'  \ottsym{=}  \ottsym{(}  \textbf{NB}  \ottsym{,}  \mathit{c}  \ottsym{,}  \ottnt{et''}  \ottsym{,}  \ottnt{vt}  \ottsym{)}  \ottsym{,}  {t_{r_{\textsf{s} } } }'  \ottsym{=}  \ottkw{clear} \, \ottnt{vt}  \ottsym{\}}\\
\ottkw{Jump} \, \mathit{r}  \ottsym{:}  \ottsym{\{}  {t_{\textsf{pc} } }  \ottsym{=}  \mathit{n}  \ottsym{,}  {t_{\textsf{ci} } }  \ottsym{=}  \ottsym{(}  \textbf{NB}  \ottsym{,}  \mathit{c}  \ottsym{,}  \ottnt{et}  \ottsym{,}  \textbf{W}  \ottsym{)}  \ottsym{,}  {t_{\textsf{ni} } }  \ottsym{=}  \ottsym{(}  \ottnt{bt'}  \ottsym{,}  \mathit{c}  \ottsym{,}  \ottnt{et'}  \ottsym{,}  \ottnt{rt'}  \ottsym{)}  \ottsym{,}  {t_{\textsf{r} } }  \ottsym{=}  \textbf{W}  \ottsym{\}}  \implies  \ottsym{\{}  {t_{\textsf{pc} } }'  \ottsym{=}  \mathit{n}  \ottsym{\}}\\
\ottkw{Jump} \, \mathit{r}  \ottsym{:}  \ottsym{\{}  {t_{\textsf{pc} } }  \ottsym{=}   \mathit{n}  + 1   \ottsym{,}  {t_{\textsf{ci} } }  \ottsym{=}  \ottsym{(}  \textbf{NB}  \ottsym{,}  \mathit{c}  \ottsym{,}  \ottnt{et}  \ottsym{,}  \textbf{W}  \ottsym{)}  \ottsym{,}  {t_{\textsf{ni} } }  \ottsym{=}  \ottsym{(}  \ottnt{bt'}  \ottsym{,}  \mathit{c'}  \ottsym{,}  \ottnt{et'}  \ottsym{,}  \ottnt{rt'}  \ottsym{)}  \ottsym{,}  {t_{\textsf{r} } }  \ottsym{=}  \ottkw{Ret} \, \mathit{n} \, \mathit{c}  \ottsym{,}  {t_{r_{ \textsf{ret} } } }  \ottsym{=}  \textbf{O} \, \mathit{c}  \ottsym{\}}  \\  \implies  \ottsym{\{}  {t_{\textsf{pc} } }'  \ottsym{=}  \mathit{n}  \ottsym{,}  {t_{\textsf{r} } }'  \ottsym{=}  \bot  \ottsym{,}  {t_{r_{ {\textsf{aux} }_1 } } }'  \ottsym{=}  \bot  \ottsym{,}  {t_{r_{ {\textsf{aux} }_2 } } }'  \ottsym{=}  \bot  \ottsym{,}  {t_{r_{ {\textsf{aux} }_3 } } }'  \ottsym{=}  \bot  \ottsym{,}  {t_{r_{ \textsf{sp} } } }'  \ottsym{=}  \bot  \ottsym{\}} \text{ when } \mathit{c}  \not=  \mathit{c'} \\
\ottkw{Jal} \, \mathit{r}  \ottsym{:}  \ottsym{\{}  {t_{\textsf{pc} } }  \ottsym{=}  \mathit{n}  \ottsym{,}  {t_{\textsf{ci} } }  \ottsym{=}  \ottsym{(}  \textbf{NB}  \ottsym{,}  \mathit{c}  \ottsym{,}  \ottnt{et}  \ottsym{,}  \textbf{W}  \ottsym{)}  \ottsym{,}  {t_{\textsf{ni} } }  \ottsym{=}  \ottsym{(}  \ottnt{bt'}  \ottsym{,}  \mathit{c}  \ottsym{,}  \ottnt{et'}  \ottsym{,}  \ottnt{rt'}  \ottsym{)}  \ottsym{,}  {t_{\textsf{r} } }  \ottsym{=}  \textbf{W}  \ottsym{,}  {t_{r_{ \textsf{a} } } }  \ottsym{=}  \ottnt{vt'}  \ottsym{\}}  \implies  \ottsym{\{}  {t_{\textsf{pc} } }'  \ottsym{=}  \mathit{n}  \ottsym{,}  {t_{r_{ \textsf{a} } } }'  \ottsym{=}  \textbf{W}  \ottsym{\}} \\
\ottkw{Jal} \, \mathit{r}  \ottsym{:}  \ottsym{\{}  {t_{\textsf{pc} } }  \ottsym{=}  \mathit{n}  \ottsym{,}  {t_{\textsf{ci} } }  \ottsym{=}  \ottsym{(}  \textbf{NB}  \ottsym{,}  \mathit{c}  \ottsym{,}  \ottnt{et}  \ottsym{,}  \textbf{W}  \ottsym{)}  \ottsym{,}  {t_{\textsf{ni} } }  \ottsym{=}  \ottsym{(}  \ottnt{bt'}  \ottsym{,}  \mathit{c'}  \ottsym{,}  \textbf{EP} \, \mathit{c_{{\mathrm{1}}}}  \rightarrow  \mathit{c_{{\mathrm{2}}}}  \ottsym{,}  \ottnt{rt'}  \ottsym{)}  \ottsym{,}  {t_{\textsf{r} } }  \ottsym{=}  \textbf{W}  \ottsym{,}  {t_{r_{ \textsf{a} } } }  \ottsym{=}  \ottnt{vt'}  \ottsym{,}  {t_{r_{ {\textsf{arg} }_1 } } }  \ottsym{=}  \textbf{O} \, \mathit{c}  \ottsym{,}  {t_{r_{ {\textsf{arg} }_2 } } }  \ottsym{=}  \textbf{O} \, \mathit{c_{{\mathrm{1}}}}  \ottsym{\}}  \\  \implies  \ottsym{\{}  {t_{\textsf{pc} } }'  \ottsym{=}   \mathit{n}  + 1   \ottsym{,}  {t_{r_{ \textsf{a} } } }'  \ottsym{=}  \ottkw{Ret} \, \mathit{n} \, \mathit{c_{{\mathrm{2}}}}  \ottsym{,}  {t_{r_{ \textsf{ret} } } }'  \ottsym{=}  \bot  \ottsym{,}  {t_{r_{ \textsf{spp} } } }'  \ottsym{=}  \bot  \ottsym{,}  {t_{r_{ \textsf{sp} } } }'  \ottsym{=}  \bot  \ottsym{\}} \text{ when } \mathit{c}  \not=  \mathit{c'} \\
\ottkw{Bnz} \, \mathit{r} \, i  \ottsym{:}  \ottsym{\{}  {t_{\textsf{pc} } }  \ottsym{=}  \mathit{n}  \ottsym{,}  {t_{\textsf{ci} } }  \ottsym{=}  \ottsym{(}  \textbf{NB}  \ottsym{,}  \mathit{c}  \ottsym{,}  \ottnt{et}  \ottsym{,}  \textbf{W}  \ottsym{)}  \ottsym{,}  {t_{\textsf{ni} } }  \ottsym{=}  \ottsym{(}  \ottnt{bt'}  \ottsym{,}  \mathit{c}  \ottsym{,}  \ottnt{et'}  \ottsym{,}  \ottnt{rt'}  \ottsym{)}  \ottsym{,}  {t_{\textsf{r} } }  \ottsym{=}  \textbf{W}  \ottsym{\}}  \implies  \ottsym{\{}  {t_{\textsf{pc} } }'  \ottsym{=}  \mathit{n}  \ottsym{\}}
\end{array}
\]
\caption{The rules of our symbolic micro-policy
}
\label{fig:micropolicy}
\end{figure*}

Our micro-policy is presented in figure \ref{fig:micropolicy},
where we use the meta notation $\ottkw{clear} \, \ottnt{vt}$ for clearing return
capabilities:
It is equal to $\ottnt{vt}$, unless $\ottnt{vt}$ is a return capability, in
which case it is equal to $\bot$.

\yj{Explanation still todo.}\yj{Not enough time to do this well}
This micro-policy combines all the informal intuition we previously
gave
in sections \ref{sec:classisolation}, \ref{sec:calldiscipline}
and \ref{sec:typesafety} into one transfer function.
A notable optimization is that we use the tag on the next instruction
to distinguish internal from cross-compartment calls and returns.
Indeed, we don't need to enforce method call discipline nor to perform
type checking on internal calls and returns.
This is a crucial move with respect to both transparency and
efficiency:
It means that low-level components don't have to comply with the
source language's abstractions for their internal computations, and
moreover this lax monitoring will result in a lower overhead on
execution time for internal steps, since there will be less cache
misses.

\paragraph{Loader Initializing Memory Tags}

The loader first performs simple static checks: it must make sure that
(1) no import declaration is left (the program is complete);
(2) there exists a $\ottkw{methl} \, \mathit{c} \, \ottmv{m}$ region for each method $\ottmv{m}$ of each
exported class $\mathit{c}$;
(3) there exists a $\ottkw{stackl} \, \mathit{c}$ region for each exported class $\mathit{c}$;
(4) there exists an $\ottkw{objl} \, \mathit{o}$ region for each exported object $\mathit{o}$;
(5) all memory regions have a matching counterpart in the export
declarations.

If all these checks succeed the loader proceeds and tags all program
memory, following the export declarations:
Every memory region is tagged uniformly with the tag of its
compartment; which is $\mathit{c}$ for $\ottkw{methl} \, \mathit{c} \, \ottmv{m}$ and $\ottkw{stackl} \, \mathit{c}$
memory regions, and the exported type of $\mathit{o}$ for $\ottkw{objl} \, \mathit{o}$
memory regions.
The first memory cell of each method region $\ottkw{methl} \, \mathit{c} \, \ottmv{m}$ gets
tagged as an entry point with the exported type signature for method
$\ottmv{m}$ of class $\mathit{c}$, while all other memory cells in the program
get tagged as not being entry points.
All locations holding an encoded $\ottkw{Const} \, \ottsym{(}  \ottkw{objl} \, \mathit{o}  \ottsym{)} \, \mathit{r}$ instruction are
tagged as blessed instructions storing a $\mathit{c}$ object pointer in
register $\mathit{r}$ ($\textbf{B} \, \mathit{c}$),
where $\mathit{c}$ is the exported type of object $\mathit{o}$.
All locations that hold a symbolic pointer $\ottkw{objl} \, \mathit{o}$
are tagged as storing a pointer to an object of class $\mathit{c}$ ($\textbf{O} \, \mathit{c}$), where
$\mathit{c}$ is the exported type of object $\mathit{o}$.
This applies to cells in both object and stack memory regions.
The other stack memory cells are tagged as cleared values $\bot$,
and all the remaining memory cells as regular words $\textbf{W}$.



\section{Related Work}
\label{sec:related}

\subsection{Secure Compilation}

Secure compilation has been the topic of many works~\cite{abadi_aslr12,
DBLP:conf/csfw/JagadeesanPRR11,patrignani2014secure,
DBLP:conf/popl/FournetSCDSL13}, but only recently
has the problem of targeting machine code been
considered~\cite{DBLP:conf/csfw/AgtenSJP12,patrignani2014secure}.
Moreover, all of these works\ch{\bf TODO: minus the recent parallel work
  from Leuvan, right?  dangerous claim, anyway}
focus on protecting a program from its
context, rather than protecting mutually distrustful components like
we do.

Abadi and Plotkin~\cite{abadi_aslr12} formalized address space layout
randomization as a secure implementation scheme for private locations,
with probabilistic guarantees:
They expressed it as a fully-abstract compilation between a source
language featuring public and private locations, and a lower-level
language in which memory addresses are natural numbers.
Follow-ups were presented by Jagadeesan \ETAL
\cite{DBLP:conf/csfw/JagadeesanPRR11} and Abadi, Planul and
Plotkin~\cite{DBLP:journals/entcs/AbadiPP13,DBLP:conf/post/AbadiP13},
with extended programming languages.

Fournet \ETAL \cite{DBLP:conf/popl/FournetSCDSL13} constructed a
fully-abstract compiler from a variant of ML~\cite{SwamyCFSBY13}
to JavaScript.
Their protection scheme protects one generated script from its
context.
A key point is that the protected script must have first-starter
privilege:
The first script which gets executed can overwrite objects in the
global namespace, on which other scripts may rely.
Hence, this scheme can't protect mutually distrustful components from
each other, since only one component can be the first to execute.

The closest work to ours is
recent~\cite{DBLP:conf/csfw/AgtenSJP12,patrignani2014secure,
patrignani_thesis} and ongoing~\cite{PatrignaniDP15,
patrignani_thesis} work by Agten, Patrignani, Piessens, \ETAL
They target a set of hardware extensions, which they call
protected module architectures~\cite{DBLP:conf/isca/McKeenABRSSS13,
DBLP:conf/isca/HoekstraLPPC13,DBLP:conf/ccs/StrackxP12}.
Micro-policies could be used to define a protected module
architecture, but they offer finer-grained protection:
In our work, this finer-grained protection allows us to manage linear
return capabilities and perform dynamic type-checking.
Micro-policies also allow us to support a stronger attacker model of dynamic
corruption, by means of a secure compilation of mutually distrustful
components.
As we discovered recently~\cite{fcs15,PatrignaniDP15},
Patrignani \ETAL are currently trying to extend there previous work to
ensure a secure compilation of mutually distrustful components.
Our hope is that this parallel work can lead to interesting comparison
and exchange, because the mechanisms we use are different:
We believe that exploiting the full power of micro-policies can
provide stronger guarantees and better performance than using
micro-policies just as an instance of protected module architectures.

\subsection{Multi-Language Approaches}

In contrast with previous fully-abstract compilers where a
single source component gets protected from its context, we protect
linked mutually distrustful low-level components from each other.

One benefit of this approach is that components need not share
a common source language.
While our current protection mechanism is still deeply connected to our
source language, in principle each component could have been
written in a specific source language and compiled using a specific
compiler.

It is actually common in real-life that the final program comes from a
mix of components that were all written in different languages.
Giving semantics to multi-language programs and building verified
compilers that can provide guarantees for them is a hot topic, studied
in particular by
Ahmed \ETAL\cite{DBLP:conf/snapl/Ahmed15,DBLP:conf/esop/PercontiA14}
and Ramananandro \ETAL\cite{DBLP:conf/cpp/RamananandroSWKF15}.

\section{Discussion and Future Work}
\label{sec:discussion}


In this section, we discuss the limitations of our work and the
generality of our approach, as well as future work.
%




\subsection{Finite Memory and Full Abstraction}

While memory is infinite in our high-level language,
memory is finite in any target low-level machine.
Our symbolic micro-policy machine is no exception:
memory regions have a fixed finite size.
This means that memory exhaustion or exposing the size of regions can
break full abstraction.

Let us first recall how memory regions are used in this work.
Our compiler translates method bodies from high-level expressions to
machine code:
Each method gets a dedicated memory region in the process, to store
its compiled code.
This code manipulates a stack that is local to the method's
compartment; and this stack also gets its own memory region.
Finally, each object gets a dedicated memory region, storing low-level
values for its fields.

The first problem is potential exhaustion of the local stack's memory:
When the stack is full and the program tries to put a new value on
top, the machine will stop.
This already breaks compiler correctness:
Executing the compiled code for $\ottsym{(}  \ottkw{this}  \ottsym{;}  \mathit{o}  \ottsym{)}$ will for example first
try to put ${r_{\textsf{tgt} } }$ on top on the full stack and hence stop the
machine, when the high-level expression would simply return $\mathit{o}$ to
the caller.
Full abstraction, which typically relies on a compiler correctness
lemma, is broken as well:
The low-level attacker can now distinguish between method bodies
$\ottsym{(}  \ottkw{this}  \ottsym{;}  \mathit{o}  \ottsym{)}$ and $\mathit{o}$, even though they have the same behavior in
the high-level.
One workaround would be to add one more intermediate step in the
compilation chain, where the symbolic machine would have infinite
memory regions:
Full abstraction would be preserved with respect to this machine, and
weakened when we move to finite memory regions.
This workaround is the one taken by CompCert~\cite{leroy09:compcert},
which until very recently~\cite{MullenTG15} only formalized and proved
something about infinite memory models.
%
A better but probably more difficult to implement solution would be to
keep the current finite-memory machine, but make it explicitly in the
property (e.g. compiler correctness, full abstraction) that in cases
such as resource exhaustion all bets are off.

The second problem is that the size of compiled method bodies, as well
as the number of private fields of compiled objects, exactly match the
size of their dedicated memory regions.
This does not cause problems with the current memory model
in which region locations exist in isolation of other regions.
In future work, switching to a more concrete view of memory
could lead to the exposure of information to the attacker:
If a program region happens to be surrounded by attacker
memory regions, then the attacker could infer the size of the program
region and hence get size information about a method body or an
object.
Because there is no similar information available in the high-level,
this will likely break full abstraction.
The concrete loader could mitigate this problem, for example by
padding memory regions so that all implementations of the same
component interface get the same size for each dedicated memory
region.
This would be, however, wasteful in practice.
Alternatively, we could weaken the property to allow leaking this
information. For instance we could weaken full abstraction to say that
one can only replace a compartment with an equivalent one that has the
same sizes when compiled. This would weaken the property quite a lot,
but it would not waste memory.  There could also be interesting
compromises in terms of security vs practicality, in which we pad to
fixed size blocks and we only leak the number of blocks.

These problems are not specific to our compiler.
Fournet \ETAL\cite{DBLP:conf/popl/FournetSCDSL13} target JavaScript
and view stack and heap memory exhaustion as a side channel of
concrete JavaScript implementations:
It is not modeled by the semantics they give to JavaScript.
Similarly, the key to the full abstraction result of Patrignani
\ETAL\cite{patrignani2014secure,patrignani2015trsem}
(the soundness and completeness of their trace semantics) is given
under the assumption that there is no overflow of the program's
stack~\cite{patrignani2015trsem}.
Patrignani \ETAL\cite{patrignani2014secure} also pad the protected
program so that for all implementations to use the same amount of
memory.

\subsection{Efficiency and Transparency}
\label{sec:efficiency}

Our micro-policy constrains low-level programs so as to prevent them
from taking potentially harmful actions.
However, we should make sure 1) that this monitoring has reasonable
impact on the performance of these programs; and 2) that these
programs are not constrained too much, in particular that benign
low-level components are not prevented from interacting with compiled
components.

\ch{Even before your first step, hardware support in the form of
  micro-policies is supposed to make enforcement more efficient than
  say doing all this in software. Alternatively mix with the next
  point by stating that we expect compartment crossing to be low
  overhead with micro-policies.}

A first, good step in this direction is that we don't enforce method
call discipline nor type safety on internal calls and returns, but
only on cross-compartment calls and returns.
This is a good idea for both efficiency and transparency:
Checks are lighter, leading to better caching and thus better
performance; and low-level programs are less constrained, while still
being prevented from taking harmful actions.

However, the constraints we set may still be too restrictive:
For example, we enforce an object-oriented view on function calls and
on data, we limit the number of arguments a function can pass through
registers, and we force the programs to comply with our type system.
\ch{I find these examples way too generally stated, in a way that no
  amount of wrapping can actually allow you to avoid these things {\em
    in full generality}. Maybe I'm wrong, but I only know one example
  of a break in transparency that's not strictly necessary, and that's
  linear return capabilities. All the rest seems to be something
  that's forced on us by the property we enforce, so no way to get
  more transparent than what we are; e.g.  in terms of type checking
  things on bounderies, number of arguments, etc. For instance what
  can a wrapper do to allow ill-typed arguments or more arguments and
  results than our functions can take?}\ch{As a fix maybe just explain
  these things by some concrete examples instead of in such general
  terms.}
This suggests the need for \emph{wrappers}.
Since internal calls and returns are not heavily monitored, we can
define methods that respect our constraints and internally call the
non-compliant benign low-level code:
This low-level code can then take its non-harmful, internal actions
without constraints --- hence with good performance --- until it
internally returns to the wrapper, which will appropriately
convert the result of the call before returning to its caller.

\subsection{Future Work}


The first crucial next step is to finish the full abstraction proof.
As we explain in section \ref{sec:attacker-model}, however, full abstraction
does not capture the exact notion of secure compilation we claim to
provide.
We will thus formalize a suitable characterization and prove it,
hopefully reusing lemmas from the full abstraction proof.
Afterwards, we will implement the compiler and conduct experiments
that could confirm or deny our hopes regarding efficiency and
transparency.

There are several ways to extend this work.
The most obvious would be to support more features that are common in
object-oriented languages, such as dynamic allocation, inheritance,
packages or exceptions.
Another way would be to move to functional languages, which provide
different, interesting challenges.
Taking as source language a lambda-calculus with references and simple
modules, would be a first step in this direction, before moving to
larger ML subsets.

Finally, the micro-policy use in this work was built progressively,
out of distinct micro-policies which we designed somewhat
independently.
Composing micro-policies in a systematic and correct way, without
breaking the guarantees of any of the composed policies, is still an
open problem that would be very interesting to study on its own.

\subsection{Scaling to Real-World Languages}

\ch{Merge all this into the Future Work subsection}

Our micro-policy seems to scale up easily to more complicated
languages, except for dynamic type checking which will be trickier.

Sub-typing, which arises with inheritance, would bring the first new
challenges in this respect.
Our dynamic type checking mechanism moreover requires encoding types
in tags:
When we move to languages with richer type systems, we will have to
explore in more detail the whole field of research on dynamic type and
contract checking~\cite{FindlerF02}.
%

Compartmentalization could easily be extended to deal with public
fields, by distinguishing memory locations that hold public field
values from other locations.

Dynamic allocation seems possible and would be handled by monitor
services, setting appropriate tags on the allocated memory region.
However, such tag initialization is expensive for large memory
regions in the current state of the PUMP, and could benefit
from additional hardware acceleration~\cite{mondrian_asplos2002}.

Finally, functional languages bring interesting challenges that have
little to do with the work presented in this document, such as closure
protection and polymorphism.
We plan to study these languages and discover how micro-policies can
help in securely compiling them.

\qed




\appendix

\section{Appendix}

\subsection{Encoding Usual Types}
\label{sec:encoding}

\begin{figure}[hb]
\centering
\begin{verbatim}
export obj decl tt : Unit
export class decl tt { }

obj tt : Unit { }
class Unit { }
\end{verbatim}
\caption{Encoding the unit type}
\label{fig:unit}
\end{figure}

\begin{figure}[b]
\centering
\begin{verbatim}
import obj decl tt : Unit
import class decl Unit { }

export obj decl t, f : Bool
export class decl Bool {
  Bool not(Unit),
  Bool and(Bool),
  Bool or(Bool)
}

obj t : Bool { }
obj f : Bool { }
class Bool {
  Bool not(Unit) { this == t ? f : t }
  Bool and(Bool) { this == t ? arg : f }
  Bool or(Bool) { this == t ? t : arg }
}
\end{verbatim}
\caption{Encoding booleans}
\label{fig:bool}
\end{figure}

\begin{figure}[b]
\centering
\begin{verbatim}
export obj decl zero, one, two, three : BNat4
export class decl BNat4 {
  BNat4 add(BNat4),
  BNat4 mul(BNat4 arg)
}

obj zero  : BNat4 { zero, one }
obj one   : BNat4 { zero, two }
obj two   : BNat4 { one,  three }
obj three : BNat4 { two,  three }
class BNat4 {
  BNat4 pred, succ;

  BNat4 add(BNat4) {
    arg == zero ?
      this : this.succ.add(arg.pred)
  }

  BNat4 mul(BNat4) {
    arg == zero ?
      zero : this.mul(arg.pred).add(this)
  }
}
\end{verbatim}
\caption{Encoding bounded natural numbers}
\label{fig:bnat}
\end{figure}

Here we give a flavor of what programming looks like with our source
language by encoding some familiar types using our class mechanism:
the unit type in figure \ref{fig:unit}, booleans in
figure \ref{fig:bool}, and bounded natural numbers in
figure~\ref{fig:bnat}.
Encoding unbounded natural numbers would be possible with dynamic
allocation, which is not part of our source language at the moment.

For better understanding, we use a syntax with strings for names which
is easily mapped to our source language syntax.
We present the three encoded types as distinct components, resulting
in quite verbose programs:
Linking them together in the high-level would result in one partial
program with three classes and no import declarations.

\subsection{Source Semantics}
\label{sec:sourcesem}

The semantics we propose for the source language is a
small-step continuation-based semantics.
It is particularly interesting to present this variant because it is
very close to our intermediate machine and can help understanding how
the source to intermediate compilation works.

After loading, source programs become a pair of a class table $\ottnt{CT}$
and an initial configuration $\ottnt{Cfg}$.
The syntax for configurations is presented on
figure~\ref{fig:sourcecfg}.

A configuration $\ottsym{(}  \ottnt{OT}  \ottsym{,}  \ottnt{CS}  \ottsym{,}  {o_{\textsf{t} } }  \ottsym{,}  {o_{\textsf{a} } }  \ottsym{,}  \ottnt{K}  \ottsym{,}  \ottnt{e}  \ottsym{)}$ can be thought of
as a machine state:
$\ottnt{OT}$ is the object table, from which we fetch field values and
which gets updated when we perform field updates.
$\ottnt{CS}$ is the call stack, on top of which we store the current
environment upon call.
${o_{\textsf{t} } }$ is the current object and ${o_{\textsf{a} } }$ the current argument.
$\ottnt{e}$ is the current expression to execute and $\ottnt{K}$ the current
continuation, which defines what we should do with the result of
evaluating $\ottnt{e}$.

Configurations can be reduced:
The rules for the reduction $\ottnt{CT}  \vdash  \ottnt{Cfg}  \longrightarrow  \ottnt{Cfg'}$ are detailed in
figure~\ref{fig:sourcesem}.
The class table $\ottnt{CT}$ is on the left of the turnstile because it
does not change throughout the computation.

The initial configuration $\ottsym{(}  \ottnt{OT}  \ottsym{,}  \text{\texttt{[]} }  \ottsym{,}  {o_{\textsf{t} } }  \ottsym{,}  {o_{\textsf{a} } }  \ottsym{,}  \text{\texttt{[]} }  \ottsym{,}  \ottnt{e}  \ottsym{)}$
features the program's object table $\ottnt{OT}$, an empty call stack, and
an empty continuation.
The current expression to execute, $\ottnt{e}$, is the body of the main
method of the program, executing with appropriately-typed current
object ${o_{\textsf{t} } }$ and argument ${o_{\textsf{a} } }$.
Since object and class names are natural numbers, an example choice
which we take in our formal study is to say that the main method is
method $0$ of class $0$, and that it should initially be called with
object $0$ of type $0$ as both current object and argument.

Our reduction is deterministic.
A program terminates with result ${o_{\textsf{r} } }$ when there is a possibly
empty reduction sequence from its initial configuration to a final
configuration $\ottsym{(}  \ottnt{OT'}  \ottsym{,}  \text{\texttt{[]} }  \ottsym{,}  {o_{\textsf{t} } }  \ottsym{,}  {o_{\textsf{a} } }  \ottsym{,}  \text{\texttt{[]} }  \ottsym{,}  {o_{\textsf{r} } }  \ottsym{)}$ or
$\ottsym{(}  \ottnt{OT'}  \ottsym{,}  \ottnt{CS}  \ottsym{,}  {o_{\textsf{t} } }  \ottsym{,}  {o_{\textsf{a} } }  \ottsym{,}  \ottsym{(}  \ottkw{exit} \, \square  \ottsym{)}  \;\textsf{::}\;  \ottnt{K}  \ottsym{,}  {o_{\textsf{r} } }  \ottsym{)}$.
The type system ensures that everywhere a $\ottkw{exit} \, \ottnt{e}$ expression is
encountered, expression $\ottnt{e}$ has the same type as the expected
return type for the main method.
Hence, when a program terminates with a value, the value
necessarily has this particular type.

\begin{figure}[ht]
\ottgrammartabular{
\ottCfg\ottinterrule
\ottCS\ottinterrule
\ottK\ottinterrule
\ottE\ottafterlastrule
}
\caption{Configuration syntax for the source language}
\label{fig:sourcecfg}
\end{figure}

\begin{figure*}[ht]
\ottgrammartabular{
\ottdefnreduce\ottafterlastrule
}
\caption{Continuation-based semantics for the source language}
\label{fig:sourcesem}
\end{figure*}


\clearpage

\bibliographystyle{plainurl}
\bibliography{safe,local}

\end{document}

%% file: rfj.ott.tex
\newcommand{\ottdrule}[4][]{{\displaystyle\frac{\begin{array}{l}#2\end{array}}{#3}\quad\ottdrulename{#4}}}
\newcommand{\ottusedrule}[1]{\[#1\]}
\newcommand{\ottpremise}[1]{ #1 \\}
\newenvironment{ottdefnblock}[3][]{ \framebox{\mbox{#2}} \quad #3 \\[0pt]}{}

\newcommand{\ottnt}[1]{\mathit{#1}}
\newcommand{\ottmv}[1]{\mathit{#1}}
\newcommand{\ottkw}[1]{\mathbf{#1}}
\newcommand{\ottsym}[1]{#1}
\newcommand{\ottcom}[1]{\text{#1}}
\newcommand{\ottdrulename}[1]{\textsc{#1}}

\newcommand{\ottgrammartabular}[1]{\begin{supertabular}{llcllllll}#1\end{supertabular}}

\newcommand{\ottrulehead}[3]{$#1$ & & $#2$ & & & \multicolumn{2}{l}{#3}}
\newcommand{\ottprodline}[6]{& & $#1$ & $#2$ & $#3 #4$ & $#5$ & $#6$}
\newcommand{\ottfirstprodline}[6]{\ottprodline{#1}{#2}{#3}{#4}{#5}{#6}}

\newcommand{\ottprodnewline}{\\}
\newcommand{\ottinterrule}{\\[5.0mm]}
\newcommand{\ottafterlastrule}{\\}

\newcommand{\ottIDT}{
\ottrulehead{\ottnt{IDT}}{::=}{\ottcom{import declaration tables}}\ottprodnewline
\ottfirstprodline{|}{\ottnt{DT}}{}{}{}{}}

\newcommand{\ottEDT}{
\ottrulehead{\ottnt{EDT}}{::=}{\ottcom{export declaration tables}}\ottprodnewline
\ottfirstprodline{|}{\ottnt{DT}}{}{}{}{}}

\newcommand{\ottDT}{
\ottrulehead{\ottnt{DT}}{::=}{\ottcom{declaration tables}}\ottprodnewline
\ottfirstprodline{|}{\ottsym{(}  \ottnt{CDT}  \ottsym{,}  \ottnt{ODT}  \ottsym{)}}{}{}{}{}}

\newcommand{\ottCDT}{
\ottrulehead{\ottnt{CDT}}{::=}{\ottcom{class declaration table}}\ottprodnewline
\ottfirstprodline{|}{\text{\texttt{[]} }}{}{}{}{}\ottprodnewline
\ottprodline{|}{\ottsym{\{}  \mathit{c}  \mapsto  \ottnt{CD}  \ottsym{\}}  \;\textsf{::}\;  \ottnt{CDT}}{}{}{}{}}

\newcommand{\ottODT}{
\ottrulehead{\ottnt{ODT}}{::=}{\ottcom{object declaration table}}\ottprodnewline
\ottfirstprodline{|}{\text{\texttt{[]} }}{}{}{}{}\ottprodnewline
\ottprodline{|}{\ottsym{\{}  \mathit{o}  \mapsto  \ottnt{OD}  \ottsym{\}}  \;\textsf{::}\;  \ottnt{ODT}}{}{}{}{}}

\newcommand{\ottCD}{
\ottrulehead{\ottnt{CD}}{::=}{\ottcom{class declaration}}\ottprodnewline
\ottfirstprodline{|}{\ottkw{class} \, \ottkw{decl} \, \ottsym{\{}  \ottnt{MD_{{\mathrm{1}}}}  \ottsym{,} \, ... \, \ottsym{,}  \ottnt{MD_{\ottmv{k}}}  \ottsym{\}}}{}{}{}{}}

\newcommand{\ottMD}{
\ottrulehead{\ottnt{MD}}{::=}{\ottcom{method declaration}}\ottprodnewline
\ottfirstprodline{|}{{c_{\textsf{r} } }  \ottsym{(}  {c_{\textsf{a} } }  \ottsym{)}}{}{}{}{}}

\newcommand{\ottOD}{
\ottrulehead{\ottnt{OD}}{::=}{\ottcom{object declaration}}\ottprodnewline
\ottfirstprodline{|}{\ottkw{obj} \, \ottkw{decl} \, \mathit{c}}{}{}{}{}}

\newcommand{\ottSP}{
\ottrulehead{P  ,\ Q}{::=}{\ottcom{source program}}\ottprodnewline
\ottfirstprodline{|}{\ottsym{(}  \ottnt{IDT}  \ottsym{,}  \ottnt{T}  \ottsym{,}  \ottnt{EDT}  \ottsym{)}}{}{}{}{}}

\newcommand{\ottT}{
\ottrulehead{\ottnt{T}}{::=}{\ottcom{definition tables}}\ottprodnewline
\ottfirstprodline{|}{\ottsym{(}  \ottnt{CT}  \ottsym{,}  \ottnt{OT}  \ottsym{)}}{}{}{}{}}

\newcommand{\ottCT}{
\ottrulehead{\ottnt{CT}}{::=}{\ottcom{class definition table}}\ottprodnewline
\ottfirstprodline{|}{\text{\texttt{[]} }}{}{}{}{}\ottprodnewline
\ottprodline{|}{\ottsym{\{}  \mathit{c}  \mapsto  \ottnt{C}  \ottsym{\}}  \;\textsf{::}\;  \ottnt{CT}}{}{}{}{}}

\newcommand{\ottOT}{
\ottrulehead{\ottnt{OT}}{::=}{\ottcom{object definition table}}\ottprodnewline
\ottfirstprodline{|}{\text{\texttt{[]} }}{}{}{}{}\ottprodnewline
\ottprodline{|}{\ottsym{\{}  \mathit{o}  \mapsto  \ottnt{O}  \ottsym{\}}  \;\textsf{::}\;  \ottnt{OT}}{}{}{}{}}

\newcommand{\ottC}{
\ottrulehead{\ottnt{C}}{::=}{\ottcom{class definition}}\ottprodnewline
\ottfirstprodline{|}{\ottkw{class} \, \ottsym{\{}  \mathit{c_{{\mathrm{1}}}}  \ottsym{,} \, ... \, \ottsym{,}  \mathit{c_{\ottmv{k}}}  \ottsym{;}  \ottnt{M_{{\mathrm{1}}}}  \ottsym{,} \, ... \, \ottsym{,}  \ottnt{M_{\ottmv{l}}}  \ottsym{\}}}{}{}{}{}}

\newcommand{\ottM}{
\ottrulehead{\ottnt{M}}{::=}{\ottcom{method definition}}\ottprodnewline
\ottfirstprodline{|}{{c_{\textsf{r} } }  \ottsym{(}  {c_{\textsf{a} } }  \ottsym{)}  \ottsym{\{}  \ottnt{e}  \ottsym{\}}}{}{}{}{}}

\newcommand{\otte}{
\ottrulehead{\ottnt{e}}{::=}{\ottcom{expressions}}\ottprodnewline
\ottfirstprodline{|}{\ottkw{this}}{}{}{}{}\ottprodnewline
\ottprodline{|}{\ottkw{arg}}{}{}{}{}\ottprodnewline
\ottprodline{|}{\mathit{o}}{}{}{}{}\ottprodnewline
\ottprodline{|}{\ottnt{e}  \ottsym{.}  \ottmv{f}}{}{}{}{}\ottprodnewline
\ottprodline{|}{\ottnt{e}  \ottsym{.}  \ottmv{f}  \ottsym{:=}  \ottnt{e'}}{}{}{}{}\ottprodnewline
\ottprodline{|}{\ottnt{e}  \ottsym{.}  \ottmv{m}  \ottsym{(}  \ottnt{e'}  \ottsym{)}}{}{}{}{}\ottprodnewline
\ottprodline{|}{\ottnt{e}  \ottsym{==}  \ottnt{e'}  \;?\;  \ottnt{e''}  \ottsym{:}  \ottnt{e'''}}{}{}{}{}\ottprodnewline
\ottprodline{|}{\ottkw{exit} \, \ottnt{e}}{}{}{}{}\ottprodnewline
\ottprodline{|}{\ottnt{e}  \ottsym{;}  \ottnt{e'}}{}{}{}{}}

\newcommand{\ottO}{
\ottrulehead{\ottnt{O}}{::=}{}\ottprodnewline
\ottfirstprodline{|}{\ottkw{obj} \, \mathit{c}  \ottsym{\{}  \mathit{o_{{\mathrm{1}}}}  \ottsym{,} \, ... \, \ottsym{,}  \mathit{o_{\ottmv{k}}}  \ottsym{\}}}{}{}{}{}}

\newcommand{\ottE}{
\ottrulehead{\ottnt{E}}{::=}{\ottcom{flat evaluation contexts}}\ottprodnewline
\ottfirstprodline{|}{\square  \ottsym{.}  \ottmv{f}}{}{}{}{}\ottprodnewline
\ottprodline{|}{\square  \ottsym{.}  \ottmv{f}  \ottsym{:=}  \ottnt{e'}}{}{}{}{}\ottprodnewline
\ottprodline{|}{\mathit{o}  \ottsym{.}  \ottmv{f}  \ottsym{:=}  \square}{}{}{}{}\ottprodnewline
\ottprodline{|}{\square  \ottsym{.}  \ottmv{m}  \ottsym{(}  \ottnt{e'}  \ottsym{)}}{}{}{}{}\ottprodnewline
\ottprodline{|}{\mathit{o}  \ottsym{.}  \ottmv{m}  \ottsym{(}  \square  \ottsym{)}}{}{}{}{}\ottprodnewline
\ottprodline{|}{\square  \ottsym{==}  \ottnt{e'}  \;?\;  \ottnt{e''}  \ottsym{:}  \ottnt{e'''}}{}{}{}{}\ottprodnewline
\ottprodline{|}{\mathit{o}  \ottsym{==}  \square  \;?\;  \ottnt{e''}  \ottsym{:}  \ottnt{e'''}}{}{}{}{}\ottprodnewline
\ottprodline{|}{\square  \ottsym{;}  \ottnt{e}}{}{}{}{}\ottprodnewline
\ottprodline{|}{\ottkw{exit} \, \square}{}{}{}{}}

\newcommand{\ottK}{
\ottrulehead{\ottnt{K}}{::=}{\ottcom{continuations}}\ottprodnewline
\ottfirstprodline{|}{\text{\texttt{[]} }}{}{}{}{}\ottprodnewline
\ottprodline{|}{\ottnt{E}  \;\textsf{::}\;  \ottnt{K}}{}{}{}{}}

\newcommand{\ottCS}{
\ottrulehead{\ottnt{CS}}{::=}{\ottcom{call stack}}\ottprodnewline
\ottfirstprodline{|}{\text{\texttt{[]} }}{}{}{}{}\ottprodnewline
\ottprodline{|}{\ottsym{(}  {o_{\textsf{t} } }  \ottsym{,}  {o_{\textsf{a} } }  \ottsym{,}  \ottnt{K}  \ottsym{)}  \;\textsf{::}\;  \ottnt{CS}}{}{}{}{}}

\newcommand{\ottCfg}{
\ottrulehead{\ottnt{Cfg}}{::=}{\ottcom{reduction configurations}}\ottprodnewline
\ottfirstprodline{|}{\ottsym{(}  \ottnt{OT}  \ottsym{,}  \ottnt{CS}  \ottsym{,}  {o_{\textsf{t} } }  \ottsym{,}  {o_{\textsf{a} } }  \ottsym{,}  \ottnt{K}  \ottsym{,}  \ottnt{e}  \ottsym{)}}{}{}{}{}}

\newcommand{\ottIP}{
\ottrulehead{\ottnt{IP}  ,\ \ottnt{IQ}}{::=}{\ottcom{intermediate program}}\ottprodnewline
\ottfirstprodline{|}{\ottsym{(}  \ottnt{IDT}  \ottsym{,}  \ottnt{ICT}  \ottsym{,}  \ottnt{EDT}  \ottsym{)}}{}{}{}{}}

\newcommand{\ottICT}{
\ottrulehead{\ottnt{ICT}}{::=}{\ottcom{compartment table}}\ottprodnewline
\ottfirstprodline{|}{\text{\texttt{[]} }}{}{}{}{}\ottprodnewline
\ottprodline{|}{\ottsym{\{}  \mathit{c}  \mapsto  \ottnt{IC}  \ottsym{\}}  \;\textsf{::}\;  \ottnt{ICT}}{}{}{}{}}

\newcommand{\ottIC}{
\ottrulehead{\ottnt{IC}}{::=}{\ottcom{compartment declaration}}\ottprodnewline
\ottfirstprodline{|}{\ottsym{\{}  \ottnt{IM_{{\mathrm{1}}}}  \ottsym{,} \, ... \, \ottsym{,}  \ottnt{IM_{\ottmv{k}}}  \ottsym{;}  \ottnt{LOT}  \ottsym{;}  \ottnt{LS}  \ottsym{\}}}{}{}{}{}}

\newcommand{\ottIM}{
\ottrulehead{\ottnt{IM}}{::=}{\ottcom{method body}}\ottprodnewline
\ottfirstprodline{|}{\ottnt{Icode}}{}{}{}{}}

\newcommand{\ottLOT}{
\ottrulehead{\ottnt{LOT}}{::=}{\ottcom{compartment local object table}}\ottprodnewline
\ottfirstprodline{|}{\text{\texttt{[]} }}{}{}{}{}\ottprodnewline
\ottprodline{|}{\ottsym{\{}  \mathit{o}  \mapsto  \ottnt{LO}  \ottsym{\}}  \;\textsf{::}\;  \ottnt{LOT}}{}{}{}{}}

\newcommand{\ottLO}{
\ottrulehead{\ottnt{LO}}{::=}{\ottcom{local instance definition}}\ottprodnewline
\ottfirstprodline{|}{\ottsym{(}  \mathit{o_{{\mathrm{1}}}}  \ottsym{,} \, ... \, \ottsym{,}  \mathit{o_{\ottmv{l}}}  \ottsym{)}}{}{}{}{}}

\newcommand{\ottLS}{
\ottrulehead{\ottnt{LS}}{::=}{\ottcom{local stack}}\ottprodnewline
\ottfirstprodline{|}{\text{\texttt{[]} }}{}{}{}{}\ottprodnewline
\ottprodline{|}{\mathit{o}  \;\textsf{::}\;  \ottnt{LS}}{}{}{}{}}

\newcommand{\ottIcode}{
\ottrulehead{\ottnt{Icode}}{::=}{\ottcom{intermediate machine code}}\ottprodnewline
\ottfirstprodline{|}{\text{\texttt{[]} }}{}{}{}{}\ottprodnewline
\ottprodline{|}{\ottnt{Iinstr}  \ottsym{;}  \ottnt{Icode}}{}{}{}{}}

\newcommand{\ottIinstr}{
\ottrulehead{\ottnt{Iinstr}}{::=}{\ottcom{machine instruction}}\ottprodnewline
\ottfirstprodline{|}{\ottkw{Nop}}{}{}{}{}\ottprodnewline
\ottprodline{|}{\ottkw{This}}{}{}{}{}\ottprodnewline
\ottprodline{|}{\ottkw{Arg}}{}{}{}{}\ottprodnewline
\ottprodline{|}{\ottkw{Ref} \, \mathit{o}}{}{}{}{}\ottprodnewline
\ottprodline{|}{\ottkw{Sel} \, \ottmv{f}}{}{}{}{}\ottprodnewline
\ottprodline{|}{\ottkw{Upd} \, \ottmv{f}}{}{}{}{}\ottprodnewline
\ottprodline{|}{\ottkw{Call} \, \mathit{c} \, \ottmv{m}}{}{}{}{}\ottprodnewline
\ottprodline{|}{\ottkw{Ret}}{}{}{}{}\ottprodnewline
\ottprodline{|}{\ottkw{Skip} \, \mathit{n}}{}{}{}{}\ottprodnewline
\ottprodline{|}{\ottkw{Skeq} \, \mathit{n}}{}{}{}{}\ottprodnewline
\ottprodline{|}{\ottkw{Drop}}{}{}{}{}\ottprodnewline
\ottprodline{|}{\ottkw{Halt}}{}{}{}{}}

\newcommand{\ottLP}{
\ottrulehead{\ottnt{LP}  ,\ \ottnt{LQ}}{::=}{\ottcom{low-level program}}\ottprodnewline
\ottfirstprodline{|}{\ottsym{(}  \ottnt{IDT}  \ottsym{,}  \ottnt{LPmem}  \ottsym{,}  \ottnt{EDT}  \ottsym{)}}{}{}{}{}}

\newcommand{\ottinstr}{
\ottrulehead{\ottnt{instr}}{::=}{\ottcom{machine instruction}}\ottprodnewline
\ottfirstprodline{|}{\ottkw{Nop}}{}{}{}{\ottcom{do nothing}}\ottprodnewline
\ottprodline{|}{\ottkw{Const} \, i \, {r_{\textsf{d} } }}{}{}{}{\ottcom{immediate value}}\ottprodnewline
\ottprodline{|}{\ottkw{Mov} \, {r_{\textsf{s} } } \, {r_{\textsf{d} } }}{}{}{}{\ottcom{copy content}}\ottprodnewline
\ottprodline{|}{ { \ottkw{Binop} _{ \ottnt{op} } }  \mathit{r_{{\mathrm{1}}}}   \mathit{r_{{\mathrm{2}}}}   {r_{\textsf{d} } } }{}{}{}{\ottcom{binary operation}}\ottprodnewline
\ottprodline{|}{\ottkw{Load} \, {r_{\textsf{p} } } \, {r_{\textsf{d} } }}{}{}{}{\ottcom{load from pointer}}\ottprodnewline
\ottprodline{|}{\ottkw{Store} \, {r_{\textsf{p} } } \, {r_{\textsf{s} } }}{}{}{}{}\ottprodnewline
\ottprodline{|}{\ottkw{Jump} \, \mathit{r}}{}{}{}{}\ottprodnewline
\ottprodline{|}{\ottkw{Jal} \, \mathit{r}}{}{}{}{}\ottprodnewline
\ottprodline{|}{\ottkw{Bnz} \, \mathit{r} \, i}{}{}{}{}\ottprodnewline
\ottprodline{|}{\ottkw{Halt}}{}{}{}{}}

\newcommand{\ottimm}{
\ottrulehead{i}{::=}{\ottcom{immediate value}}\ottprodnewline
\ottfirstprodline{|}{\mathit{n}}{}{}{}{}\ottprodnewline
\ottprodline{|}{\ottnt{loc}  \ottsym{+}  \mathit{n}}{}{}{}{}}

\newcommand{\ottloc}{
\ottrulehead{\ottnt{loc}}{::=}{\ottcom{region symbolic location}}\ottprodnewline
\ottfirstprodline{|}{\ottkw{objl} \, \mathit{o}}{}{}{}{}\ottprodnewline
\ottprodline{|}{\ottkw{methl} \, \mathit{c} \, \ottmv{m}}{}{}{}{}\ottprodnewline
\ottprodline{|}{\ottkw{stackl} \, \mathit{c}}{}{}{}{}}

\newcommand{\ottmem}{
\ottrulehead{\ottnt{mem}}{::=}{\ottcom{memory}}\ottprodnewline
\ottfirstprodline{|}{\text{\texttt{[]} }}{}{}{}{}\ottprodnewline
\ottprodline{|}{\ottsym{\{}  \ottnt{loc}  \mapsto  \ottnt{R}  \ottsym{\}}  \;\textsf{::}\;  \ottnt{mem}}{}{}{}{}}

\newcommand{\ottR}{
\ottrulehead{\ottnt{R}}{::=}{\ottcom{tagged memory region}}\ottprodnewline
\ottfirstprodline{|}{\text{\texttt{[]} }}{}{}{}{}\ottprodnewline
\ottprodline{|}{ (  \ottnt{word}   \;\textsf{@}\;   {t_{\textsf{mem} } }  )  \;\textsf{::}\;   \ottnt{R} }{}{}{}{}}

\newcommand{\ottLPmem}{
\ottrulehead{\ottnt{LPmem}}{::=}{\ottcom{program memory}}\ottprodnewline
\ottfirstprodline{|}{\text{\texttt{[]} }}{}{}{}{}\ottprodnewline
\ottprodline{|}{\ottsym{\{}  \ottnt{loc}  \mapsto  \ottnt{LPR}  \ottsym{\}}  \;\textsf{::}\;  \ottnt{LPmem}}{}{}{}{}}

\newcommand{\ottLPR}{
\ottrulehead{\ottnt{LPR}}{::=}{\ottcom{program memory region}}\ottprodnewline
\ottfirstprodline{|}{\text{\texttt{[]} }}{}{}{}{}\ottprodnewline
\ottprodline{|}{\ottnt{word}  \;\textsf{::}\;  \ottnt{LPR}}{}{}{}{}}

\newcommand{\ottword}{
\ottrulehead{\ottnt{word}}{::=}{\ottcom{symbolic machine word}}\ottprodnewline
\ottfirstprodline{|}{\mathit{n}}{}{}{}{}\ottprodnewline
\ottprodline{|}{\ottnt{loc}  \ottsym{+}  \mathit{n}}{}{}{}{}\ottprodnewline
\ottprodline{|}{\ottkw{encode} \, \ottnt{instr}}{}{}{}{}}

\newcommand{\otttpc}{
\ottrulehead{{t_{\textsf{pc} } }}{::=}{\ottcom{program counter tag}}\ottprodnewline
\ottfirstprodline{|}{\mathit{n}}{}{}{}{\ottcom{current call depth}}}

\newcommand{\otttmem}{
\ottrulehead{{t_{\textsf{mem} } }}{::=}{\ottcom{memory tag}}\ottprodnewline
\ottfirstprodline{|}{\ottsym{(}  \ottnt{bt}  \ottsym{,}  \mathit{c}  \ottsym{,}  \ottnt{et}  \ottsym{,}  \ottnt{vt}  \ottsym{)}}{}{}{}{}}

\newcommand{\ottet}{
\ottrulehead{\ottnt{et}}{::=}{\ottcom{entry point tag}}\ottprodnewline
\ottfirstprodline{|}{\textbf{EP} \, {c_{\textsf{a} } }  \rightarrow  {c_{\textsf{r} } }}{}{}{}{}\ottprodnewline
\ottprodline{|}{\textbf{NEP}}{}{}{}{}}

\newcommand{\ottbt}{
\ottrulehead{\ottnt{bt}}{::=}{\ottcom{blessed tag}}\ottprodnewline
\ottfirstprodline{|}{\textbf{B} \, \mathit{c}}{}{}{}{}\ottprodnewline
\ottprodline{|}{\textbf{NB}}{}{}{}{}}

\newcommand{\otttreg}{
\ottrulehead{{t_{\textsf{reg} } }}{::=}{\ottcom{register tag}}\ottprodnewline
\ottfirstprodline{|}{\ottnt{vt}}{}{}{}{}}

\newcommand{\ottvt}{
\ottrulehead{\ottnt{vt}}{::=}{\ottcom{regular or cleared value tag}}\ottprodnewline
\ottfirstprodline{|}{\ottnt{rt}}{}{}{}{}\ottprodnewline
\ottprodline{|}{\bot}{}{}{}{}}

\newcommand{\ottrt}{
\ottrulehead{\ottnt{rt}}{::=}{\ottcom{regular value tag}}\ottprodnewline
\ottfirstprodline{|}{\ottkw{Ret} \, \mathit{n} \, {c_{\textsf{r} } }}{}{}{}{\ottcom{return capability}}\ottprodnewline
\ottprodline{|}{\textbf{O} \, \mathit{c}}{}{}{}{\ottcom{object pointer tagged with its type}}\ottprodnewline
\ottprodline{|}{\textbf{W}}{}{}{}{}}

\newcommand{\ottdrulethis}[1]{\ottdrule[#1]{%
}{
\ottnt{CT}  \vdash  \ottsym{(}  \ottnt{OT}  \ottsym{,}  \ottnt{CS}  \ottsym{,}  {o_{\textsf{t} } }  \ottsym{,}  {o_{\textsf{a} } }  \ottsym{,}  \ottnt{K}  \ottsym{,}  \ottkw{this}  \ottsym{)}  \longrightarrow  \ottsym{(}  \ottnt{OT}  \ottsym{,}  \ottnt{CS}  \ottsym{,}  {o_{\textsf{t} } }  \ottsym{,}  {o_{\textsf{a} } }  \ottsym{,}  \ottnt{K}  \ottsym{,}  {o_{\textsf{t} } }  \ottsym{)}}{%
{\ottdrulename{this}}{}%
}}

\newcommand{\ottdrulearg}[1]{\ottdrule[#1]{%
}{
\ottnt{CT}  \vdash  \ottsym{(}  \ottnt{OT}  \ottsym{,}  \ottnt{CS}  \ottsym{,}  {o_{\textsf{t} } }  \ottsym{,}  {o_{\textsf{a} } }  \ottsym{,}  \ottnt{K}  \ottsym{,}  \ottkw{arg}  \ottsym{)}  \longrightarrow  \ottsym{(}  \ottnt{OT}  \ottsym{,}  \ottnt{CS}  \ottsym{,}  {o_{\textsf{t} } }  \ottsym{,}  {o_{\textsf{a} } }  \ottsym{,}  \ottnt{K}  \ottsym{,}  {o_{\textsf{a} } }  \ottsym{)}}{%
{\ottdrulename{arg}}{}%
}}

\newcommand{\ottdruleselXXpush}[1]{\ottdrule[#1]{%
}{
\ottnt{CT}  \vdash  \ottsym{(}  \ottnt{OT}  \ottsym{,}  \ottnt{CS}  \ottsym{,}  {o_{\textsf{t} } }  \ottsym{,}  {o_{\textsf{a} } }  \ottsym{,}  \ottnt{K}  \ottsym{,}  \ottnt{e}  \ottsym{.}  \ottmv{f}  \ottsym{)}  \longrightarrow  \ottsym{(}  \ottnt{OT}  \ottsym{,}  \ottnt{CS}  \ottsym{,}  {o_{\textsf{t} } }  \ottsym{,}  {o_{\textsf{a} } }  \ottsym{,}  \square  \ottsym{.}  \ottmv{f}  \;\textsf{::}\;  \ottnt{K}  \ottsym{,}  \ottnt{e}  \ottsym{)}}{%
{\ottdrulename{sel\_push}}{}%
}}

\newcommand{\ottdruleselXXpop}[1]{\ottdrule[#1]{%
\ottpremise{\ottnt{OT}  \ottsym{(}  {o_{\textsf{t} } }  \ottsym{)}  \ottsym{=}  \ottkw{obj} \, {c_{\textsf{t} } }  \ottsym{\{}  {o_{\textsf{t} } }_{{\mathrm{1}}}  \ottsym{,} \, ... \, \ottsym{,}  {o_{\textsf{t} } }_{\ottmv{k}}  \ottsym{\}}}%
\ottpremise{\ottnt{OT}  \ottsym{(}  \mathit{o}  \ottsym{)}  \ottsym{=}  \ottkw{obj} \, {c_{\textsf{t} } }  \ottsym{\{}  \mathit{o_{{\mathrm{1}}}}  \ottsym{,} \, ... \, \ottsym{,}  \mathit{o_{\ottmv{l}}}  \ottsym{\}}}%
\ottpremise{\ottsym{1}  \leq  \ottmv{f}  \leq  \ottmv{l}}%
}{
\ottnt{CT}  \vdash  \ottsym{(}  \ottnt{OT}  \ottsym{,}  \ottnt{CS}  \ottsym{,}  {o_{\textsf{t} } }  \ottsym{,}  {o_{\textsf{a} } }  \ottsym{,}  \square  \ottsym{.}  \ottmv{f}  \;\textsf{::}\;  \ottnt{K}  \ottsym{,}  \mathit{o}  \ottsym{)}  \longrightarrow  \ottsym{(}  \ottnt{OT}  \ottsym{,}  \ottnt{CS}  \ottsym{,}  {o_{\textsf{t} } }  \ottsym{,}  {o_{\textsf{a} } }  \ottsym{,}  \ottnt{K}  \ottsym{,}  \mathit{o_{\ottmv{f}}}  \ottsym{)}}{%
{\ottdrulename{sel\_pop}}{}%
}}

\newcommand{\ottdruleupdXXpush}[1]{\ottdrule[#1]{%
}{
\ottnt{CT}  \vdash  \ottsym{(}  \ottnt{OT}  \ottsym{,}  \ottnt{CS}  \ottsym{,}  {o_{\textsf{t} } }  \ottsym{,}  {o_{\textsf{a} } }  \ottsym{,}  \ottnt{K}  \ottsym{,}  \ottnt{e}  \ottsym{.}  \ottmv{f}  \ottsym{:=}  \ottnt{e'}  \ottsym{)}  \longrightarrow  \ottsym{(}  \ottnt{OT}  \ottsym{,}  \ottnt{CS}  \ottsym{,}  {o_{\textsf{t} } }  \ottsym{,}  {o_{\textsf{a} } }  \ottsym{,}  \ottsym{(}  \square  \ottsym{.}  \ottmv{f}  \ottsym{:=}  \ottnt{e'}  \ottsym{)}  \;\textsf{::}\;  \ottnt{K}  \ottsym{,}  \ottnt{e}  \ottsym{)}}{%
{\ottdrulename{upd\_push}}{}%
}}

\newcommand{\ottdruleupdXXswitch}[1]{\ottdrule[#1]{%
}{
\ottnt{CT}  \vdash  \ottsym{(}  \ottnt{OT}  \ottsym{,}  \ottnt{CS}  \ottsym{,}  {o_{\textsf{t} } }  \ottsym{,}  {o_{\textsf{a} } }  \ottsym{,}  \ottsym{(}  \square  \ottsym{.}  \ottmv{f}  \ottsym{:=}  \ottnt{e}  \ottsym{)}  \;\textsf{::}\;  \ottnt{K}  \ottsym{,}  \mathit{o}  \ottsym{)}  \longrightarrow  \ottsym{(}  \ottnt{OT}  \ottsym{,}  \ottnt{CS}  \ottsym{,}  {o_{\textsf{t} } }  \ottsym{,}  {o_{\textsf{a} } }  \ottsym{,}  \ottsym{(}  \mathit{o}  \ottsym{.}  \ottmv{f}  \ottsym{:=}  \square  \ottsym{)}  \;\textsf{::}\;  \ottnt{K}  \ottsym{,}  \ottnt{e}  \ottsym{)}}{%
{\ottdrulename{upd\_switch}}{}%
}}

\newcommand{\ottdruleupdXXpop}[1]{\ottdrule[#1]{%
\ottpremise{\ottnt{OT}  \ottsym{(}  {o_{\textsf{t} } }  \ottsym{)}  \ottsym{=}  \ottkw{obj} \, {c_{\textsf{t} } }  \ottsym{\{}  {o_{\textsf{t} } }_{{\mathrm{1}}}  \ottsym{,} \, ... \, \ottsym{,}  {o_{\textsf{t} } }_{\ottmv{k}}  \ottsym{\}}}%
\ottpremise{\ottnt{OT}  \ottsym{(}  \mathit{o}  \ottsym{)}  \ottsym{=}  \ottkw{obj} \, {c_{\textsf{t} } }  \ottsym{\{}  \mathit{o_{{\mathrm{1}}}}  \ottsym{,} \, ... \, \ottsym{,}  \mathit{o_{{\ottmv{f}-1}}}  \ottsym{,}  \mathit{o_{\ottmv{f}}}  \ottsym{,}  \mathit{o'_{\ottmv{j}}}  \ottsym{,} \, ... \, \ottsym{,}  \mathit{o'_{\ottmv{l}}}  \ottsym{\}}}%
\ottpremise{\ottnt{OT'}  \ottsym{=}  \ottnt{OT}  \ottsym{[}  \mathit{o}  \mapsto  \ottkw{obj} \, \mathit{c}  \ottsym{\{}  \mathit{o_{{\mathrm{1}}}}  \ottsym{,} \, ... \, \ottsym{,}  \mathit{o_{{\ottmv{f}-1}}}  \ottsym{,}  \mathit{o''}  \ottsym{,}  \mathit{o'_{\ottmv{j}}}  \ottsym{,} \, ... \, \ottsym{,}  \mathit{o'_{\ottmv{l}}}  \ottsym{\}}  \ottsym{]}}%
}{
\ottnt{CT}  \vdash  \ottsym{(}  \ottnt{OT}  \ottsym{,}  \ottnt{CS}  \ottsym{,}  {o_{\textsf{t} } }  \ottsym{,}  {o_{\textsf{a} } }  \ottsym{,}  \ottsym{(}  \mathit{o}  \ottsym{.}  \ottmv{f}  \ottsym{:=}  \square  \ottsym{)}  \;\textsf{::}\;  \ottnt{K}  \ottsym{,}  \mathit{o''}  \ottsym{)}  \longrightarrow  \ottsym{(}  \ottnt{OT'}  \ottsym{,}  \ottnt{CS}  \ottsym{,}  {o_{\textsf{t} } }  \ottsym{,}  {o_{\textsf{a} } }  \ottsym{,}  \ottnt{K}  \ottsym{,}  \mathit{o''}  \ottsym{)}}{%
{\ottdrulename{upd\_pop}}{}%
}}

\newcommand{\ottdrulecallXXpush}[1]{\ottdrule[#1]{%
}{
\ottnt{CT}  \vdash  \ottsym{(}  \ottnt{OT}  \ottsym{,}  \ottnt{CS}  \ottsym{,}  {o_{\textsf{t} } }  \ottsym{,}  {o_{\textsf{a} } }  \ottsym{,}  \ottnt{K}  \ottsym{,}  \ottnt{e}  \ottsym{.}  \ottmv{m}  \ottsym{(}  \ottnt{e'}  \ottsym{)}  \ottsym{)}  \longrightarrow  \ottsym{(}  \ottnt{OT}  \ottsym{,}  \ottnt{CS}  \ottsym{,}  {o_{\textsf{t} } }  \ottsym{,}  {o_{\textsf{a} } }  \ottsym{,}  \square  \ottsym{.}  \ottmv{m}  \ottsym{(}  \ottnt{e'}  \ottsym{)}  \;\textsf{::}\;  \ottnt{K}  \ottsym{,}  \ottnt{e}  \ottsym{)}}{%
{\ottdrulename{call\_push}}{}%
}}

\newcommand{\ottdrulecallXXswitch}[1]{\ottdrule[#1]{%
}{
\ottnt{CT}  \vdash  \ottsym{(}  \ottnt{OT}  \ottsym{,}  \ottnt{CS}  \ottsym{,}  {o_{\textsf{t} } }  \ottsym{,}  {o_{\textsf{a} } }  \ottsym{,}  \square  \ottsym{.}  \ottmv{m}  \ottsym{(}  \ottnt{e}  \ottsym{)}  \;\textsf{::}\;  \ottnt{K}  \ottsym{,}  \mathit{o}  \ottsym{)}  \longrightarrow  \ottsym{(}  \ottnt{OT}  \ottsym{,}  \ottnt{CS}  \ottsym{,}  {o_{\textsf{t} } }  \ottsym{,}  {o_{\textsf{a} } }  \ottsym{,}  \mathit{o}  \ottsym{.}  \ottmv{m}  \ottsym{(}  \square  \ottsym{)}  \;\textsf{::}\;  \ottnt{K}  \ottsym{,}  \ottnt{e}  \ottsym{)}}{%
{\ottdrulename{call\_switch}}{}%
}}

\newcommand{\ottdrulecallXXpop}[1]{\ottdrule[#1]{%
\ottpremise{\ottnt{OT}  \ottsym{(}  {o_{\textsf{t} } }'  \ottsym{)}  \ottsym{=}  \ottkw{obj} \, {c_{\textsf{t} } }'  \ottsym{\{}  {o_{\textsf{t} } }'_{{\mathrm{1}}}  \ottsym{,} \, ... \, \ottsym{,}  {o_{\textsf{t} } }'_{\ottmv{k}}  \ottsym{\}}}%
\ottpremise{\ottnt{CT}  \ottsym{(}  {c_{\textsf{t} } }'  \ottsym{)}  \ottsym{=}  \ottkw{class} \, \ottsym{\{}  \mathit{c_{{\mathrm{1}}}}  \ottsym{,} \, ... \, \ottsym{,}  \mathit{c_{\ottmv{i}}}  \ottsym{;}  \ottnt{M_{{\mathrm{1}}}}  \ottsym{,} \, ... \, \ottsym{,}  \ottnt{M_{\ottmv{j}}}  \ottsym{\}}}%
\ottpremise{\ottsym{1}  \leq  \ottmv{m}  \leq  \ottmv{j}}%
\ottpremise{\ottnt{M_{\ottmv{m}}}  \ottsym{=}  {c_{\textsf{r} } }  \ottsym{(}  {c_{\textsf{a} } }  \ottsym{)}  \ottsym{\{}  \ottnt{e}  \ottsym{\}}}%
}{
\ottnt{CT}  \vdash  \ottsym{(}  \ottnt{OT}  \ottsym{,}  \ottnt{CS}  \ottsym{,}  {o_{\textsf{t} } }  \ottsym{,}  {o_{\textsf{a} } }  \ottsym{,}  {o_{\textsf{t} } }'  \ottsym{.}  \ottmv{m}  \ottsym{(}  \square  \ottsym{)}  \;\textsf{::}\;  \ottnt{K}  \ottsym{,}  {o_{\textsf{a} } }'  \ottsym{)}  \longrightarrow  \ottsym{(}  \ottnt{OT}  \ottsym{,}  \ottsym{(}  {o_{\textsf{t} } }  \ottsym{,}  {o_{\textsf{a} } }  \ottsym{,}  \ottnt{K}  \ottsym{)}  \;\textsf{::}\;  \ottnt{CS}  \ottsym{,}  {o_{\textsf{t} } }'  \ottsym{,}  {o_{\textsf{a} } }'  \ottsym{,}  \ottnt{K}  \ottsym{,}  \ottnt{e}  \ottsym{)}}{%
{\ottdrulename{call\_pop}}{}%
}}

\newcommand{\ottdrulereturn}[1]{\ottdrule[#1]{%
}{
\ottnt{CT}  \vdash  \ottsym{(}  \ottnt{OT}  \ottsym{,}  \ottsym{(}  {o_{\textsf{t} } }'  \ottsym{,}  {o_{\textsf{a} } }'  \ottsym{,}  \ottnt{K}  \ottsym{)}  \;\textsf{::}\;  \ottnt{CS}  \ottsym{,}  {o_{\textsf{t} } }  \ottsym{,}  {o_{\textsf{a} } }  \ottsym{,}  \text{\texttt{[]} }  \ottsym{,}  {o_{\textsf{r} } }  \ottsym{)}  \longrightarrow  \ottsym{(}  \ottnt{OT}  \ottsym{,}  \ottnt{CS}  \ottsym{,}  {o_{\textsf{t} } }'  \ottsym{,}  {o_{\textsf{a} } }'  \ottsym{,}  \ottnt{K}  \ottsym{,}  {o_{\textsf{r} } }  \ottsym{)}}{%
{\ottdrulename{return}}{}%
}}

\newcommand{\ottdruletestXXpush}[1]{\ottdrule[#1]{%
}{
\ottnt{CT}  \vdash  \ottsym{(}  \ottnt{OT}  \ottsym{,}  \ottnt{CS}  \ottsym{,}  {o_{\textsf{t} } }  \ottsym{,}  {o_{\textsf{a} } }  \ottsym{,}  \ottnt{K}  \ottsym{,}  \ottnt{e_{{\mathrm{1}}}}  \ottsym{==}  \ottnt{e_{{\mathrm{2}}}}  \;?\;  \ottnt{e_{{\mathrm{3}}}}  \ottsym{:}  \ottnt{e_{{\mathrm{4}}}}  \ottsym{)}  \longrightarrow  \ottsym{(}  \ottnt{OT}  \ottsym{,}  \ottnt{CS}  \ottsym{,}  {o_{\textsf{t} } }  \ottsym{,}  {o_{\textsf{a} } }  \ottsym{,}  \ottsym{(}  \square  \ottsym{==}  \ottnt{e_{{\mathrm{2}}}}  \;?\;  \ottnt{e_{{\mathrm{3}}}}  \ottsym{:}  \ottnt{e_{{\mathrm{4}}}}  \ottsym{)}  \;\textsf{::}\;  \ottnt{K}  \ottsym{,}  \ottnt{e_{{\mathrm{1}}}}  \ottsym{)}}{%
{\ottdrulename{test\_push}}{}%
}}

\newcommand{\ottdruletestXXswitch}[1]{\ottdrule[#1]{%
}{
\ottnt{CT}  \vdash  \ottsym{(}  \ottnt{OT}  \ottsym{,}  \ottnt{CS}  \ottsym{,}  {o_{\textsf{t} } }  \ottsym{,}  {o_{\textsf{a} } }  \ottsym{,}  \ottsym{(}  \square  \ottsym{==}  \ottnt{e_{{\mathrm{2}}}}  \;?\;  \ottnt{e_{{\mathrm{3}}}}  \ottsym{:}  \ottnt{e_{{\mathrm{4}}}}  \ottsym{)}  \;\textsf{::}\;  \ottnt{K}  \ottsym{,}  \mathit{o_{{\mathrm{1}}}}  \ottsym{)}  \longrightarrow  \ottsym{(}  \ottnt{OT}  \ottsym{,}  \ottnt{CS}  \ottsym{,}  {o_{\textsf{t} } }  \ottsym{,}  {o_{\textsf{a} } }  \ottsym{,}  \ottsym{(}  \mathit{o_{{\mathrm{1}}}}  \ottsym{==}  \square  \;?\;  \ottnt{e_{{\mathrm{3}}}}  \ottsym{:}  \ottnt{e_{{\mathrm{4}}}}  \ottsym{)}  \;\textsf{::}\;  \ottnt{K}  \ottsym{,}  \ottnt{e_{{\mathrm{2}}}}  \ottsym{)}}{%
{\ottdrulename{test\_switch}}{}%
}}

\newcommand{\ottdruletestXXpopXXeq}[1]{\ottdrule[#1]{%
}{
\ottnt{CT}  \vdash  \ottsym{(}  \ottnt{OT}  \ottsym{,}  \ottnt{CS}  \ottsym{,}  {o_{\textsf{t} } }  \ottsym{,}  {o_{\textsf{a} } }  \ottsym{,}  \ottsym{(}  \mathit{o_{{\mathrm{1}}}}  \ottsym{==}  \square  \;?\;  \ottnt{e_{{\mathrm{3}}}}  \ottsym{:}  \ottnt{e_{{\mathrm{4}}}}  \ottsym{)}  \;\textsf{::}\;  \ottnt{K}  \ottsym{,}  \mathit{o_{{\mathrm{1}}}}  \ottsym{)}  \longrightarrow  \ottsym{(}  \ottnt{OT}  \ottsym{,}  \ottnt{CS}  \ottsym{,}  {o_{\textsf{t} } }  \ottsym{,}  {o_{\textsf{a} } }  \ottsym{,}  \ottnt{K}  \ottsym{,}  \ottnt{e_{{\mathrm{3}}}}  \ottsym{)}}{%
{\ottdrulename{test\_pop\_eq}}{}%
}}

\newcommand{\ottdruletestXXpopXXneq}[1]{\ottdrule[#1]{%
\ottpremise{\mathit{o_{{\mathrm{1}}}}  \not=  \mathit{o_{{\mathrm{2}}}}}%
}{
\ottnt{CT}  \vdash  \ottsym{(}  \ottnt{OT}  \ottsym{,}  \ottnt{CS}  \ottsym{,}  {o_{\textsf{t} } }  \ottsym{,}  {o_{\textsf{a} } }  \ottsym{,}  \ottsym{(}  \mathit{o_{{\mathrm{1}}}}  \ottsym{==}  \square  \;?\;  \ottnt{e_{{\mathrm{3}}}}  \ottsym{:}  \ottnt{e_{{\mathrm{4}}}}  \ottsym{)}  \;\textsf{::}\;  \ottnt{K}  \ottsym{,}  \mathit{o_{{\mathrm{2}}}}  \ottsym{)}  \longrightarrow  \ottsym{(}  \ottnt{OT}  \ottsym{,}  \ottnt{CS}  \ottsym{,}  {o_{\textsf{t} } }  \ottsym{,}  {o_{\textsf{a} } }  \ottsym{,}  \ottnt{K}  \ottsym{,}  \ottnt{e_{{\mathrm{4}}}}  \ottsym{)}}{%
{\ottdrulename{test\_pop\_neq}}{}%
}}

\newcommand{\ottdruleseqXXpush}[1]{\ottdrule[#1]{%
}{
\ottnt{CT}  \vdash  \ottsym{(}  \ottnt{OT}  \ottsym{,}  \ottnt{CS}  \ottsym{,}  {o_{\textsf{t} } }  \ottsym{,}  {o_{\textsf{a} } }  \ottsym{,}  \ottnt{K}  \ottsym{,}  \ottsym{(}  \ottnt{e}  \ottsym{;}  \ottnt{e'}  \ottsym{)}  \ottsym{)}  \longrightarrow  \ottsym{(}  \ottnt{OT}  \ottsym{,}  \ottnt{CS}  \ottsym{,}  {o_{\textsf{t} } }  \ottsym{,}  {o_{\textsf{a} } }  \ottsym{,}  \ottsym{(}  \square  \ottsym{;}  \ottnt{e'}  \ottsym{)}  \;\textsf{::}\;  \ottnt{K}  \ottsym{,}  \ottnt{e}  \ottsym{)}}{%
{\ottdrulename{seq\_push}}{}%
}}

\newcommand{\ottdruleseqXXpop}[1]{\ottdrule[#1]{%
}{
\ottnt{CT}  \vdash  \ottsym{(}  \ottnt{OT}  \ottsym{,}  \ottnt{CS}  \ottsym{,}  {o_{\textsf{t} } }  \ottsym{,}  {o_{\textsf{a} } }  \ottsym{,}  \ottsym{(}  \square  \ottsym{;}  \ottnt{e}  \ottsym{)}  \;\textsf{::}\;  \ottnt{K}  \ottsym{,}  \mathit{o}  \ottsym{)}  \longrightarrow  \ottsym{(}  \ottnt{OT}  \ottsym{,}  \ottnt{CS}  \ottsym{,}  {o_{\textsf{t} } }  \ottsym{,}  {o_{\textsf{a} } }  \ottsym{,}  \ottnt{K}  \ottsym{,}  \ottnt{e}  \ottsym{)}}{%
{\ottdrulename{seq\_pop}}{}%
}}

\newcommand{\ottdruleexitXXpush}[1]{\ottdrule[#1]{%
}{
\ottnt{CT}  \vdash  \ottsym{(}  \ottnt{OT}  \ottsym{,}  \ottnt{CS}  \ottsym{,}  {o_{\textsf{t} } }  \ottsym{,}  {o_{\textsf{a} } }  \ottsym{,}  \ottnt{K}  \ottsym{,}  \ottkw{exit} \, \ottnt{e}  \ottsym{)}  \longrightarrow  \ottsym{(}  \ottnt{OT}  \ottsym{,}  \ottnt{CS}  \ottsym{,}  {o_{\textsf{t} } }  \ottsym{,}  {o_{\textsf{a} } }  \ottsym{,}  \ottsym{(}  \ottkw{exit} \, \square  \ottsym{)}  \;\textsf{::}\;  \ottnt{K}  \ottsym{,}  \ottnt{e}  \ottsym{)}}{%
{\ottdrulename{exit\_push}}{}%
}}

\newcommand{\ottdefnreduce}[1]{\begin{ottdefnblock}[#1]{$\ottnt{CT}  \vdash  \ottnt{Cfg}  \longrightarrow  \ottnt{Cfg'}$}{}
\ottusedrule{\ottdrulethis{}}
\ottusedrule{\ottdrulearg{}}
\ottusedrule{\ottdruleselXXpush{}}
\ottusedrule{\ottdruleselXXpop{}}
\ottusedrule{\ottdruleupdXXpush{}}
\ottusedrule{\ottdruleupdXXswitch{}}
\ottusedrule{\ottdruleupdXXpop{}}
\ottusedrule{\ottdrulecallXXpush{}}
\ottusedrule{\ottdrulecallXXswitch{}}
\ottusedrule{\ottdrulecallXXpop{}}
\ottusedrule{\ottdrulereturn{}}
\ottusedrule{\ottdruletestXXpush{}}
\ottusedrule{\ottdruletestXXswitch{}}
\ottusedrule{\ottdruletestXXpopXXeq{}}
\ottusedrule{\ottdruletestXXpopXXneq{}}
\ottusedrule{\ottdruleseqXXpush{}}
\ottusedrule{\ottdruleseqXXpop{}}
\ottusedrule{\ottdruleexitXXpush{}}
\end{ottdefnblock}}

